\PassOptionsToPackage{prologue,dvipsnames}{xcolor}
\documentclass[11pt]{article}
\usepackage{lmodern}
\usepackage{amsmath}
\usepackage{amssymb}
\usepackage{mathtools}
\usepackage{amsthm}

\usepackage{booktabs} % For formal tables
\usepackage[ruled]{algorithm2e} % For algorithms

\SetAlFnt{\small}
\SetAlCapFnt{\small}
\SetAlCapNameFnt{\small}
\SetAlCapHSkip{0pt}
\IncMargin{-\parindent}
\usepackage[square]{natbib}

\setcitestyle{authoryear}
\usepackage{authblk}
\usepackage{titlesec} % For title formatting

\title{\begin{flushleft}\LARGE\textbf{Method of Equal Shares with Bounded Overspending}\end{flushleft}}

\newcommand{\WAS}{W\k{A}S}
\newcommand{\Was}{W\k{a}s}

\author{\begin{flushleft}\large GEORGIOS PAPASOTIROPOULOS, \textit{\small University of Warsaw}\\ \vspace{-0.2cm}
		\large SEYEDEH ZEINAB PISHBIN, \textit{\small University of Warsaw, University of Tehran}\\ \vspace{0.15cm}
		\large OSKAR SKIBSKI, \textit{\small University of Warsaw}\\ \vspace{0.15cm}
		\large PIOTR SKOWRON, \textit{\small University of Warsaw}\\ \vspace{0.15cm}
		\large TOMASZ \WAS, \textit{\small University of Oxford}
	\end{flushleft}
}

\usepackage{geometry}

\usepackage{utopia} %nicer font

\usepackage{tikz}
\usetikzlibrary{patterns,decorations.pathreplacing,arrows.meta,shapes.geometric}

\usepackage{nicefrac}
\usepackage{paralist}
\usepackage[shortlabels]{enumitem}
\usepackage{xcolor}
\usepackage{colortbl}
\usepackage{subcaption}
\usepackage{listings}
\usepackage{fontawesome}
\usepackage{multirow}
\usepackage{tabularx}
\usepackage{pgfplots}
\pgfplotsset{compat=1.18}
\usepackage{pifont}

\definecolor{winered}{rgb}{0.5,0.1,0.1}
\definecolor{darkgreen}{rgb}{0,0.5,0}
\renewcommand{\faCheck}{\textcolor{darkgreen}{\ding{51}}}
\newcolumntype{Y}{>{\centering\arraybackslash}X}
\newcolumntype{Z}{>{\raggedleft\arraybackslash}X}
\newcolumntype{a}{>{\columncolor{black!05}}Y}
\newcommand{\rightcell}[1]{\multicolumn{1}{r}{#1} }

\renewcommand*{\le}{\leqslant}
\renewcommand*{\leq}{\leqslant}
\renewcommand*{\ge}{\geqslant}
\renewcommand*{\geq}{\geqslant}
\renewcommand{\epsilon}{\varepsilon}

\newcommand{\reals}{{{\mathbb{R}}}}

\newcommand{\cost}{{{\mathrm{cost}}}}

\newcommand{\ini}{{{\mathrm{ini}}}}

\renewcommand{\part}{{{\mathrm{part}}}}

\newcommand{\calR}{{{\mathcal{R}}}}
\newcommand{\dist}{{{\mathrm{dist}}}}

\usepackage{thmtools,thm-restate}

\usepackage{setspace}

\renewcommand*{\le}{\leqslant}
\renewcommand*{\leq}{\leqslant}
\renewcommand*{\ge}{\geqslant}
\renewcommand*{\geq}{\geqslant}
\renewcommand{\epsilon}{\varepsilon}

\usepackage[T1]{fontenc}

\newcommand{\myemph}[1]{{\color{winered}\emph{#1}}}
\newcommand{\costc}{{{\mathrm{cost(c)}}}}
\renewcommand{\part}{{{\mathrm{part}}}}

\newcommand\supminus{\smash{\scalebox{0.55}[0.55]{\(-\)}}}
\newcommand\supplus{\smash{\scalebox{0.55}[0.55]{\(+\)}}}

\newtheoremstyle{ex}
{}{}
{}%{\itshape}
{}
{\bfseries}
{.}
{ }
{%
	\thmname{#1}% the label
	\thmnumber{ #2}% the number
	\thmnote{{\bfseries: \ifex(\fi#3\ifex) \fi}}% the note
}
\newif\ifex

{\theoremstyle{ex}{
		\newtheorem{exof}{Example}}}

\newcommand{\argmin}{{\mathrm{argmin}}}

\newcommand{\lv}[1]{}
\newcommand{\sv}[1]{#1}
\usepackage{etoolbox}
\usepackage{tablefootnote}
\newcommand{\appendixText}{}
\newcommand{\toappendix}[1]{\gappto{\appendixText}{{#1}}}

\usepackage{titlesec}

\newcommand{\appendixsectionformat}{%
	\titleformat{\section}[block]
	{\normalfont\large\bfseries}{Appendix \Alph{section}.}{1em}{}
}

\usepackage{tocloft}
\usepackage{multicol}
\usepackage[colorlinks=true, linkcolor=winered, citecolor=darkgreen, urlcolor=blue]{hyperref}
\usepackage[sort&compress,nameinlink]{cleveref}
\Crefname{table}{Table}{Tables}
\Crefname{table}{Table}{Tables}
\Crefname{figure}{Figure}{Figures}
\Crefname{theorem}{Theorem}{Theorems}
\Crefname{definition}{Definition}{Definitions}
\Crefname{corollary}{Corollary}{Corollaries}
\Crefname{observation}{Observation}{Observations}
\Crefname{question}{Question}{Question}
\Crefname{lemma}{Lemma}{Lemmas}
\Crefname{claim}{Claim}{Claims}
\Crefname{example}{Example}{Examples}
\Crefname{reduction}{Reduction}{Reductions}
\Crefname{construction}{Construction}{Constructions}
\Crefname{subsection}{Section}{Sections}
\Crefname{section}{Section}{Sections}
\Crefname{proposition}{Proposition}{Propositions}
\Crefname{algorithm}{Algorithm}{Algorithms}
\Crefname{algocf}{Algorithm}{Algorithms}
\Crefname{equation}{Inequality}{Inequalities}
\Crefname{lstlisting}{listing}{listings}
\Crefname{exof}{Example}{Examples}
\theoremstyle{definition}
\newtheorem{definition}{Definition}
\newtheorem{example}{Example}
\newtheorem*{examplecont}{Continuation of \Cref{ex:mes}}

\theoremstyle{plain}
\newtheorem{theorem}{Theorem}

\newtheorem{corollary}[theorem]{Corollary}

\newtheorem{claim}[theorem]{Claim}

\renewenvironment{description}
{\list{}{\labelwidth=8pt \leftmargin=12pt
		}}
{\endlist}

\newcommand{\pabulib}{{{\textsc{Pabulib}}}}

\begin{document}

	\date{}

	\newgeometry{left=0.95in, right=0.95in, top=-0.05in, bottom=0.8in}
	\maketitle

	\vspace{-1cm}{\small\noindent	In participatory budgeting (PB), voters decide through voting which subset of projects to fund within a given budget. 
		Proportionality in the context of PB is crucial to ensure equal treatment of all groups of voters. 
		However, pure proportional rules can sometimes lead to suboptimal outcomes. We introduce the Method of Equal Shares with Bounded Overspending (BOS Equal Shares), a robust variant of Equal Shares that balances proportionality and efficiency. 
		BOS Equal Shares addresses inefficiencies implied by strict proportionality axioms, yet the rule still provides fairness guarantees, similar to the original Method of Equal Shares. 
		Our extensive empirical analysis on real-world PB instances shows excellent performance of BOS Equal Shares across several metrics.
		In the course of the analysis, we also present and examine a fractional variant of the Method of Equal Shares which allows for partial funding of projects.}
	
	\setcounter{tocdepth}{2} % adjust to 1 if desired
	%\begin{multicols}{2} 
	\noindent\rule{\textwidth}{0.8pt}
	\begin{center}
		\textsc{Contents}  % Centered and larger font
	\end{center}
	\vspace{-1.2cm}
	\renewcommand{\contentsname}{}
	\begin{spacing}{0.9}
		\tableofcontents  % Generate TOC
		%\end{multicols}
	\end{spacing}
	\noindent\rule{\textwidth}{0.5pt}
	
	\setcounter{secnumdepth}{2}
	
	\newgeometry{left=0.95in, right=0.95in, top=0.75in, bottom=0.8in}
	\section{Introduction}

We consider the participatory budgeting (PB) scenario, where a group of voters decides, through voting, which subset of projects to fund. 
The projects have varying prices, and the total cost of the selected projects cannot exceed a given budget. 
Voters express their preferences by casting ballots---typically either by indicating sets of approved projects (i.e., they cast approval ballots) or by assigning numerical scores to projects (so-called range voting). 
PB has been recently adopted by many municipalities worldwide,\footnote{\url{https://en.wikipedia.org/wiki/List_of_participatory_budgeting_votes}} but the model applies more broadly. 
It extends the framework of committee elections~\citep{lac-sko:multiwinner-book, FSST-trends} and can be used even if the voters and candidates are not humans but represent abstract objects, such as validators in proof-of-stake blockchains~\citep{cevallos2020validator}.

Proportionality is a critical requirement in the context of PB elections. 
Intuitively, it says that each group of similar-minded voters should be entitled to decide about a proportional fraction of the available funds (e.g., if 30\% of voters like similar projects, then roughly  30\% of funds should be designated to the projects these voters support). 
Proportionality, among others, ensures equal treatment of minorities, geographical regions, and various project categories~\citep{fal-fli-pet-pie-sko-sto-szu-tal:pb-experiments, rey-mal:survey-on-pb}. 
It also ensures that groups of voters forming pluralities are not overrepresented, thus protecting elections against certain strategies employed by coordinated voters or project owners. 
As a result, several proportionality criteria and new voting rules have been proposed in the literature (cf. the overview of \citet{rey-mal:survey-on-pb}). 
One voting method, the Method of Equal Shares~\citep{pet-sko:laminar,pet-pie-sko:c:participatory-budgeting-cardinal}, stands out by exhibiting particularly strong proportionality properties~\citep{los2022proportional,brill2023proportionality} as well as robustness to changes in the voter participation~\citep{benade2023participatory}.
It has also been successfully used in real-life PB elections.%
\footnote{\url{https://equalshares.net/elections}}

Informally speaking, the Method of Equal Shares first virtually distributes the budget equally among the voters. 
Projects are then selected based on the vote count, but each time a project is selected, its cost is split among its supporting voters. 
Thus, only the projects whose supporters still have enough funds to cover their costs can be selected.
Furthermore, the votes of those who have run out of money are no longer counted. 
This way, in subsequent rounds, the votes of minorities who have not yet influenced the decision (hence, still have money) are taken into account.

While this method offers strong proportionality guarantees, and largely exhibits desired behavior in practice,
it is not without its flaws.
In this work, we begin by uncovering drawbacks of Equal Shares (discussed in detail in \Cref{sec:challenging-instances,sec:experiments}).

\begin{description}
	\item[Underspending]
	As it has been already observed in the literature, the Method of Equal Shares can significantly underspend the available funds~\citep{fal-fli-pet-pie-sko-sto-szu-tal:pb-experiments}.
	Indicatively, on real instances of PB elections it uses on average only 45\% of the budget (see \Cref{sec:experiments}).
	Thus, in practice it heavily depends on the completion method with which it is combined (cf. \Cref{sec:prelims}).
	\item[Helenka Paradox]
	In Equal Shares, even a small group of voters can propose a modest project which, if unanimously supported by that group, would likely be selected. However, due to strict budget constraints, this could prevent a larger project---potentially benefiting the vast majority of voters---from being funded. We observed this issue in the Helenka district of Zabrze, Poland, where the Method of Equal Shares would have resulted in 97\% of voters left with no project.
	\item[Tail Utilities] 
	When there are large discrepancies in voters' scores assigned to projects, Equal Shares selects a project based on the score of the least satisfied voter who must cover its cost. This mechanism  is egalitarian in principle, and so, the method may not properly take into account significant utility values in the decision process.
\end{description}

Based on these observations, we introduce and analyze a new voting rule, which we call the \myemph{Method of Equal Shares with Bounded Overspending} (in short, BOS Equal Shares, or simply BOS). 
The new method is a more robust variant of Equal Shares, remarkably effective in handling scenarios akin to particularly challenging instances of PB elections.

Interestingly, we argue that some of the problems we identify are not solely due to the Method of Equal Shares itself, but rather to the strict proportionality guarantees the method aims to provide. 
Specifically, it appears that the appealing axiom of extended justified representation (EJR)~\citep{aziz2017justified, pet-pie-sko:c:participatory-budgeting-cardinal, BP-ejrplus} might enforce inefficiencies in certain scenarios. 
Consequently, BOS Equal Shares is not meant to generally satisfy EJR. 
Nevertheless, in most cases, our new method provides strong proportionality guarantees, mirroring those of the original Method of Equal Shares. 
To confirm this we prove that BOS satisfies an approximate variant of EJR.

Our main argument for the advantages of BOS Equal Shares
comes from our extensive empirical analysis on real-world PB instances,
in which BOS shows very good and robust performance in a number of metrics.
In particular, it provides EJR+ up to one (a strong EJR-style axiom proposed by \citet{BP-ejrplus}) in more than 95\% of cases,
and, in comparison to the original Method of Equal Shares, leaves less voters empty-handed.
Furthermore, our rule has been recently proved to be superior in the context of selecting a representative set of influential nodes in networks~\citep{pap-ski-sko-was25}.

Noteworthy, in the course of designing BOS Equal Shares, we propose and analyze a fractional variant of the Method of Equal Shares, which we refer to as the \myemph{Fractional Equal Shares} (FrES). 
The new variant works in a model, where projects are allowed to be funded partially. It extends the Generalized Method of Equal Shares~\citep{lu2024approval}, another rule recently proposed for the fractional model, but only for approval ballots. 

\section{Preliminaries}
\label{sec:prelims}
A PB election (in short, an election), is a tuple $E = (C, V, b)$, where $C = \{c_1, \ldots, c_m\}$ is a set of available \myemph{candidates} (also referred to as \myemph{projects}), $V = \{v_1, \ldots, v_n\}$ is a set of voters, and $b \in \mathbb{R}$ is the budget value. 
Each candidate $c \in C$ is associated with a cost, denoted as $\cost(c),$ assumed to be upper bounded by $b$. 
We extend this notation to sets of candidates, setting $\cost(W) = \sum_{c \in W} \cost(c)$ for all $W \subseteq C$. 
An \myemph{outcome} of an election is a subset of candidates; an outcome $W$ is \myemph{feasible} if $\cost(W) \leq b$. An \myemph{election rule} is a function that for each election returns a nonempty collection of feasible outcomes. 
Typically, we are interested in a single outcome, yet we allow for ties. 
Each voter $v_i\in V$ has a utility function $u_i \colon C \to \reals_{\geq 0}$ that assigns values to the candidates. 
We assume the utilities are additive, and write $u_i(W) = \sum_{c \in W} u_i(c)$ for each $W \subseteq C$. 
The voters' utility functions and the candidates' costs are integral parts of the election.

There are two special types of utility functions that are particularly interesting, and pivotal for certain parts of our work. In some cases it is natural to assume that the utilities directly correspond to the scores extracted from the voters' ballots; then we speak about \myemph{score utilities}.
In case of approval ballots, where the voters only indicate subsets of supported candidates, we simply assume that  the voter assigns scores of one to the approved candidates, and  scores of zero to those she does not approve.
On the other hand, the \myemph{cost utility} of a voter from a project is its score utility multiplied by the projects' cost. 
For example, if $W \subseteq C$ and $A_i $ is the set of candidates that voter $v_i$ approves, then we have $u_i(W) = |W \cap A_i|$ for score utilities (the utility is the number of selected projects the voter approves), and $u_i(W) = \cost(W \cap A_i)$ for cost utilities (thus, the utility is the amount of public funds allocated to projects she supports).

\subsection{Method of Equal Shares}
\label{sec:mes}

Arguably the simplest, and most commonly used voting rule in PB elections is the Utilitarian Method.
This method selects the candidates by their vote count, omitting those whose selection would exceed the budget; it stops when no further candidate can be added. 
Since this approach is highly suboptimal from the perspective of proportionality (see for instance \Cref{ex:mes} that follows), we will be particularly interested in the Method of Equal Shares, a proportional election rule recently introduced in the literature~\citep{pet-pie-sko:c:participatory-budgeting-cardinal, pet-sko:laminar}.

\begin{description}
	\item[Method of Equal Shares] 
	Let $b_i$ be the virtual budget of voter $v_i$; initially $b_i := b_{\ini} = \nicefrac{b}{n}$.
	In each round, we say that a not yet elected project $c$ is \myemph{$\rho$-affordable} for $\rho \in \mathbb{R}_{+}$, if
	\begin{align*}
		\textstyle \cost(c) = \sum_{v_i \in V} \min\left(b_i, u_i(c) \cdot\rho\right).
	\end{align*}
	In a given round the method selects the $\rho$-affordable candidate for the lowest possible value of $\rho$ and updates the voters' accounts accordingly: $b_i := b_i - \min\left(b_i, u_i(c) \cdot\rho\right)$; then it moves to the next round.
	The rule stops if there is no $\rho$-affordable candidate for any value of $\rho$. 
\end{description}

The concept of $\rho$-affordability is crucial to our work. Intuitively, voters supporting a $\rho$-affordable candidate $c$ can cover its cost in such a way that each of them pays $\rho$ per unit of utility or all of their remaining funds. In simpler terms, $\rho$ represents the rate (price per unit of utility) at which the least advantaged supporter of the project would “purchase'' their satisfaction, if the project is selected.
Also, note that a situation in which a candidate is not $\rho$-affordable for all $\rho \in \mathbb{R}_{+}$ happens only if its supporters do not have enough money to cover its cost.

\begin{example}\label{ex:mes}
	
	\begin{table}[t]
		\centering
		\begin{tabularx}{12cm}{lc|XXXXXXXXXX}
			& cost & $v_1$  & $v_2$ & $v_3$ & $v_4$ & $v_5$ & $v_6$ & $v_7$ & $v_8$ & $v_9$ & $v_{10}$ \\
			\midrule
			Project A 
			& \$300k & \faCheck & \faCheck & \faCheck & \faCheck & \faCheck & \faCheck & & & & \\
			Project B  
			& \$400k & & \faCheck & \faCheck & \faCheck & \faCheck & \faCheck & & & & \\
			Project C  
			& \$300k & & \faCheck & \faCheck & \faCheck  & \faCheck  & & & & & \faCheck \\
			Project D 
			& \$240k & & & & & & & \faCheck & \faCheck & \faCheck & \faCheck \\
			Project E  
			& \$170k & & \faCheck & & & & & \faCheck & \faCheck & \faCheck & \\
			Project F  
			& \$100k & & & & & & \faCheck & & & \faCheck & \faCheck  \\
		\end{tabularx}
		\caption{An example of a PB election with $10$ voters and $6$ projects of varying costs. Approvals of voters towards projects are indicated by the {\faCheck} symbol. Given a budget of $\$1{,}000{,}000$, the Utilitarian Method selects $\{\text{A}, \text{B}, \text{C}\}$, the Method of Equal Shares selects $\{\text{A},\text{D},\text{E}\}$, while the BOS Equal Shares selects $\{\text{A},\text{C},\text{D},\text{F}\}$.}
		\label{tab:running_example}	
	\end{table}

	Consider the PB election depicted in \Cref{tab:running_example}, and assume cost-utilities.
	The Utilitarian Method would select projects solely based on their vote count, thus choosing Projects A, B, and C. 
	This seems unfair since a large fraction of the voters (namely voters $v_7$ to $v_9$ making up $30\%$ of the electorate) would not approve any of the selected projects.

	At the beginning, the Method of Equal Shares assigns \$100k to each voter. The table that appears below indicates the available (virtual) budget of each voter (in thousands of dollars).
	\begin{center}
		\setlength{\tabcolsep}{2pt}
		\begin{tabularx}{10cm}{c|YYYYYYYYYY}
			& $v_1$ & $v_2$ & $v_3$ & $v_4$ & $v_5$ & $v_6$ & $v_7$ & $v_8$ & $v_9$ & $v_{10}$ \\
			\midrule
			\ $b_i$\ \ & \small $100$ & \small $100$ & \small $100$ & \small $100$ & \small $100$ & \small $100$ & \small $100$ & \small $100$ & \small $100$ & \small $100$ \\
		\end{tabularx}
	\end{center}
	We first need to determine how affordable each project is. Project A is $\nicefrac{1}{6}$-affordable, as it can be funded if each of its supporters pays \$50k, which is $\nicefrac{1}{6}$ of its cost (recall that we assumed cost-utilities for this example).
	Analogously, each other project that received $x$ votes is $\nicefrac{1}{x}$-affordable.
	\begin{center}
		\setlength{\tabcolsep}{2pt}
		\begin{tabularx}{5cm}{c|aYYYYY}
			& A & B & C & D & E & F \\
			\midrule
			\ $\rho$\ \  & \nicefrac{1}{6} & \nicefrac{1}{5} & \nicefrac{1}{5} & \nicefrac{1}{4} & \nicefrac{1}{4} & \nicefrac{1}{3} \\
		\end{tabularx}
	\end{center}
	Thus, in the first round the rule simply selects the project with the highest vote count, namely project A. After paying its cost, voters $v_1$ to $v_6$ are left with $\$(100\text{k} - 300\text{k}/6) = \$50\text{k}$. Voters' remaining budget follows.
	\begin{center}
		\setlength{\tabcolsep}{2pt}
		\begin{tabularx}{10cm}{c|YYYYYYYYYY}
			& $v_1$ & $v_2$ & $v_3$ & $v_4$ & $v_5$ & $v_6$ & $v_7$ & $v_8$ & $v_9$ & $v_{10}$ \\
			\midrule
			\ $b_i$\ \ & \small $50$ & \small $50$ & \small $50$ & \small $50$ & \small $50$ & \small $50$ & \small $100$ & \small $100$ & \small $100$ & \small $100$ \\
		\end{tabularx}
	\end{center}
	
	In the second round, project B is no longer affordable, as its supporters do not have a total of at least \$400k to fund it. Project C is $\nicefrac{1}{3}$-affordable. Indeed, to fund it, voters $v_2$, $v_3$, $v_4$, $v_5$ and $v_{10}$ would have to use all their money. This means, in particular, that voter $v_{10}$ would pay \$100k out of \$300k which is $\nicefrac{1}{3}$ of the cost of the project. The rest of the projects remain $\rho$-affordable for the same values of $\rho$, as their costs can still be spread equally among their supporters. 
	\begin{center}
		\setlength{\tabcolsep}{2pt}
		\begin{tabularx}{5cm}{c|aYYaYY}
			& A & B & C & D & E & F \\
			\midrule
			\ $\rho$\ \ & -- & -- & \nicefrac{1}{3} & \nicefrac{1}{4} & \nicefrac{1}{4} & \nicefrac{1}{3} \\
		\end{tabularx}
	\end{center}
	Hence, projects D and E are both $\rho$-affordable for the smallest value of $\rho = \nicefrac{1}{4}$. Let us assume the former project is selected, by breaking ties lexicographically. After paying its cost, voters $v_7$ to $v_{10}$ are left with $\$(100\text{k} - 240\text{k}/4) = \$40\text{k}$.
	\begin{center}
		\setlength{\tabcolsep}{2pt}
		\begin{tabularx}{10cm}{c|YYYYYYYYYY}
			& $v_1$ & $v_2$ & $v_3$ & $v_4$ & $v_5$ & $v_6$ & $v_7$ & $v_8$ & $v_9$ & $v_{10}$ \\
			\midrule
			\ $b_i$\ \  & \small $50$ & \small $50$ & \small $50$ & \small $50$ & \small $50$ & \small $50$ & \small $40$ & \small $40$ & \small $40$ & \small $40$ \\
		\end{tabularx}
	\end{center}
	
	In the third round, project E is selected with $\rho = \nicefrac{5}{17}$ (see \Cref{app:examples} for details),
	which makes voters $v_2$, $v_7$, $v_8$, and $v_9$ run out of money.
	\begin{center}
		\setlength{\tabcolsep}{2pt}
		\begin{tabularx}{10cm}{c|YYYYYYYYYY}
			& $v_1$ & $v_2$ & $v_3$ & $v_4$ & $v_5$ & $v_6$ & $v_7$ & $v_8$ & $v_9$ & $v_{10}$ \\
			\midrule
			\ $b_i$\ \  & \small $50$ & \small $0$ & \small $50$ & \small $50$ & \small $50$ & \small $50$ & \small $0$ & \small $0$ & \small $0$ & \small $40$ \\
		\end{tabularx}
	\end{center}
	
	At this point, no project is affordable since the supporters of projects B, C and F have in total \$200k, \$150k and \$90k, respectively, therefore, the procedure stops having selected the outcome $\{A,D,E\}$.
	Note that the purchased bundle comes at a total cost of \$710k, which is \$290k less than the initially available budget.
	Thus, in principle, we could afford to additionally fund project F.
	However, the supporters of this project do not have enough (virtual) money to fund it,
	and so the project is not selected by the Method of Equal Shares. 
	
	Clearly, the selection made by the Method of Equal Shares is less discriminatory than the one by Utilitarian,
	as each voter approves at least one of the selected projects. 
	\hfill $\lrcorner$
\end{example}

\Cref{ex:mes} shows that Equal Shares is non-exhaustive: an outcome $W \subseteq C$ is \myemph{exhaustive} if it utilizes the available funds in a way that no further project can be funded, in other words if for each unelected candidate $c \notin W$ it holds that $\cost(W \cup \{c\}) > b$. While non-exhaustiveness itself may not be a critical flaw, a more concerning issue is that the method tends to significantly underspend the available funds. In real instances of participatory budgeting elections, Equal Shares allocates, on average, only 45\% of the available budget (see \Cref{sec:experiments}).

To deal with this issue, Equal Shares is typically used together with a completion strategy. 
An example of a well-performing strategy suggested in the literature is Add1U:
\begin{description}
	\item[Add1U] 
	We gradually increment the initial endowment $b_{\ini}$ by one unit
	and rerun the Method of Equal Shares from scratch, until it produces an outcome that exceeds the budget.
	Then we return the outcome computed for the previous value of $b_\ini$, hence the feasible outcome produced for the highest tested value of $b_\ini$. 
	Since the result may still be non-exhaustive (though typically at this point most of the funds are already spent), as the final step, we select affordable projects with the highest vote count until no further candidate can be added
	\citep{fal-fli-pet-pie-sko-sto-szu-tal:pb-experiments}.
\end{description}

\section{Limitations of the Method of Equal Shares}
\label{sec:challenging-instances}

In this section, we present concrete case studies which indicate that using the Method of Equal Shares in its basic form may lead to intuitively suboptimal solutions. 

\subsection{Helenka Paradox}
\label{ex:helenka-paradox}

The instance that follows comes from the PB elections held in 2020 in the Polish city of Zabrze, in district Helenka.%
\footnote{\url{https://pabulib.org/?search=Poland\%20Zabrze\%202021\%20Helenka}}
Two projects were proposed in this district, namely an expansion and modernization of sports facilities (to be called project A), and a plant sculpture (project B). 
Their costs and number of supporters are as follows.

\begin{center}
	\begin{tabularx}{8cm}{lr|YY}
		& cost & 403 voters  & 11 voters \\
		\midrule
		Project A & \$310k & \faCheck &  \\
		Project B  & \$6k & & \faCheck  \\
	\end{tabularx}
\end{center}

The budget is $b = \$310,000$ and assume cost utilities. 
The second group of 11 voters should intuitively be entitled to $\$(\nicefrac{11}{414} \cdot 310,000) \approx \$8,000$. 
In fact, any rule that satisfies EJR must select project~B and so does Equal Shares.	
However, selecting project~B precludes the inclusion of project~A within the budget constraint. 
This is highly counterintuitive since it leaves a great majority of voters (over 97\% of the electorate) empty handed, despite the fact that they commonly approve an affordable project. Thus, the Helenka Paradox serves not only as a critique of the Method of Equal Shares, but also of the prominent axiom of EJR itself.

In order to solve the indicated problem, additional strategies could be employed. 
Cities may compare the outcomes returned by the Method of Equal Shares and by the standard Utilitarian Method. 
If ballots indicate that more voters prefer the outcome of the standard Utilitarian Method, then it could be selected.
An alternative solution is to put an upper bound on the initial cost of the projects. 
Interestingly, while working on the experimental part of our work, we observed no evident paradoxes like the discussed one in data from elections where the cost of the projects did not exceed 30\% of the budget. 
While these two solutions can work well in practice, neither is perfect. 
The runoff approach might potentially result in a utilitarian solution where some groups of voters are underrepresented.
Moreover, imposing an upper limit on project costs might exclude some worthwhile and highly popular ideas, especially in small-scale elections.

\subsection{Tail Utilities}
\label{ex:tail-utils}

Consider the following election with $m=2$ candidates, and $n=100$ voters casting ballots via range voting, as follows (the values in the table indicate the assigned scores).

\begin{center}
	\begin{tabularx}{7.5cm}{lr|YY}
		& cost & 99 voters  & 1 voter \\
		\midrule
		Project A & \$1 & 100 & 1 \\
		Project B & \$1 & 2 & 2  \\
	\end{tabularx}
\end{center}

Assume that the budget is $b = \$1$. 
Under the Method of Equal Shares, all voters have to pay all their virtual money to cover the cost of the one project that will be selected.
As a result, if project A is selected, 99 voters will pay $\$0.0001$ per unit of utility and 1 voter will pay $\$0.01$ per unit of utility.
In turn, for project B, all voters will pay $\$0.005$ per unit of utility.
Thus, project A is $\nicefrac{1}{100}$-affordable while project B is $\nicefrac{1}{200}$-affordable,
and the rule selects project B, even though 99\% of voters consider project A as a much better option.
This is because the Method of Equal Shares is in some sense egalitarian: when assessing the quality of a candidate,
it essentially considers the utility assigned to the candidate by the least satisfied voter among those covering its cost. 

Note that the presented problem does not appear in approval elections. 
We observed this issue when applying Equal Shares to certain range voting committee elections.

\section{Method of Fractional Equal Shares}
\label{sec:fres}

To develop intuitions required for the introduction of BOS Equal Shares,
we first present Fractional Equal Shares (FrES)---an adaptation of Equal Shares
to fractional PB,
where projects can be partially funded.
The idea is simple: a fraction $\alpha$ of a candidate $c$ can be bought for the corresponding fraction of its cost: $\alpha \cdot \cost(c)$. 
If such a purchase is made, each voter $v_i \in V$ receives the utility of $\alpha \cdot u_i(c)$. 

Let us start by extending the notion of $\rho$-affordability to project fractions.
For $\alpha \in (0,1]$, we say that a candidate $c$ is \myemph{$(\alpha,\rho)$-affordable} if its $\alpha$ fraction can be bought with ratio $\rho$ such that
\begin{equation}
\label{eq:alpha:rho:affordable}
	\textstyle \alpha \cdot \cost(c) = \sum_{v_i \in V} \min(b_i, \alpha \cdot u_i(c) \cdot \rho).
\end{equation}
% We assume that a project is always $(\alpha,\rho)$-affordable, for $\rho = +\infty$ and any value of $\alpha$.

\begin{description}
	\item[Fractional Equal Shares] 
	Let $b_i$ be the virtual budget of voter $v_i$, initially set to $b_i := \nicefrac{b}{n}$.
	In each round,
	the method selects the candidate $c$ which is $(\alpha,\rho)$-affordable for the lowest possible value of $\rho$ and buys the largest possible fraction $\alpha$ for which the candidate remains $(\alpha,\rho)$-affordable.
	Then, the accounts of the supporters of $c$ are updated accordingly: $b_i := b_i - \min(b_i, \alpha \cdot u_i(c) \cdot \rho)$.
	The method stops when no further fraction of a project can be selected within the budget.
\end{description}

Let us explain this method in more detail. Consider a partially funded candidate $c$ and let $S$ be the set of its supporters who still have money. Note that as $\alpha$ increases, the ratio $\rho$ cannot decrease, meaning the minimum value of $\rho$ occurs at small $\alpha$. If $\alpha$ is small enough, the candidate can be funded with all voters contributing proportionally to their utilities, since no one exhausts their funds,
i.e., in $\min(b_i, \alpha \cdot u_i(c) \cdot \rho)$, it is always $\alpha \cdot u_i(c) \cdot \rho$ that is smaller (or equal). Then, from Eq.~\eqref{eq:alpha:rho:affordable} we obtain that
\[ \rho = \frac{\cost(c)}{\sum_{v_i \in S} u_i(c)}. \]
Hence, Fractional Equal Shares selects the project with the lowest such $\rho$ value and covers the fraction of the candidate's cost with payments proportional to the voters' utilities. This fraction is determined by the first moment a supporter exhausts their funds or the selected project is fully funded, whichever comes first. 
Hence, the rule runs in polynomial time. 
Moreover,
\[ \alpha = \min \left(1 - W_c, \min_{v_i \in S} \frac{b_i}{\rho \cdot u_i(c)} \right), \]
where $W_c$ is the fraction of the project $c$ bought already. 
% Since the payments of voters from $V$ are proportional to their utility, they are equal to $\alpha \cdot u_i(c) \cdot \rho$ for each voter $v_i \in S$.
\Cref{algorithm:fres} contains the pseudo-code of the Fractional Equal Shares.
% When voters who still have money support only candidates which are already fully bought, all other candidates have $\rho = +\infty$.
% In the pseudo-code, we break ties according to the total utility of all voters divided by the cost of the project and buy as large fraction of them as possible.

\begin{algorithm}[t]
	\SetKw{Return}{return}
	\SetKw{Input}{Input:}
	\SetKw{Break}{break}
	\captionsetup{labelfont={sc,bf}, labelsep=newline}
	\Input{A PB election $(C,V,b)$\\}
	\DontPrintSemicolon
	\SetAlgoNoEnd
	\SetAlgoLined
	$W_c \gets 0 \text{~for each~} c \in C$ \; %. 
	$b_i \gets \nicefrac{b}{n}  \text{~for each~} v_i \in V$ \; %
	$S \gets V$ \;
	\While{\emph{exists} $c$ \emph{such that} $W_c \neq 1$ \emph{and} $\sum_{v_i \in S} u_i(c) > 0$}{
		$c \gets \argmin_{c \in C : W_c \neq 1} \left(\cost(c)/\sum_{v_i \in S} u_i(c)\right)$ \; 
		$\rho \gets \ \cost(c)/\sum_{v_i \in S} u_i(c)$ \;
		$\alpha \gets \min(1-W_c, \min_{v_i \in S} b_i / (\rho \cdot u_i(c)))$\;
		$W_{c} \gets W_{c} + \alpha$ \;
		\For{$v_i \in S$}{
			$b_i \gets b_i - \alpha \cdot \rho \cdot u_i(c)$ \;
			\If{$b_i = 0$}{
				$S \gets S \setminus \{v_i\}$ \;
			}	
		}
	}

	\Return{W}\;
	\caption{Pseudo-code of Fractional Equal Shares.}
	\label{algorithm:fres}
\end{algorithm}

\begin{examplecont}
	As an illustration, we review \Cref{ex:mes} this time applying Fractional Equal Shares.
	Again, we assume cost-utilities.
	Note that in such a setting, in each round, FrES selects the project
	with the most supporters who still have money left.
	The method purchases the largest portion of the project that can be covered with equal payments of the supporters.
	
	In the first round, as in the Method of Equal Shares, project A is chosen.
	It is bought in full, as its whole price can be split equally between supporters, each paying \$50k.
The voters' remaining funds can be seen in the table that follows.
	\begin{center}
		\setlength{\tabcolsep}{2pt}
		\begin{tabularx}{10cm}{c|YYYYYYYYYY}
			& $v_1$ & $v_2$ & $v_3$ & $v_4$ & $v_5$ & $v_6$ & $v_7$ & $v_8$ & $v_9$ & $v_{10}$ \\
			\midrule
			\ $b_i$\ \ & \small $50$ & \small $50$ & \small $50$ & \small $50$ & \small $50$ & \small $50$ & \small $100$ & \small $100$ & \small $100$ & \small $100$ \\
		\end{tabularx}
	\end{center}
	
	In the second round, project B or C can be selected. Let us assume project C is selected. Since voters $v_2$ to $v_5$ have only \$50k left, only $\nicefrac{5}{6}$ of it is bought and each voter pays \$50k.
	\begin{center}
		\setlength{\tabcolsep}{2pt}
		\begin{tabularx}{6cm}{c|aYaYYY}
			& A & B & C & D & E & F \\
			\midrule
			\ $\rho$\ \  & -- & \nicefrac{1}{5} & \nicefrac{1}{5} & \nicefrac{1}{4} & \nicefrac{1}{4} & \nicefrac{1}{3} \\
			\midrule
			\ $\alpha$\ \  & 1 & & $\nicefrac{5}{6}$ \\
		\end{tabularx}
	\end{center}
	The following table depicts the remaining amount of money of each voter.
	\begin{center}
		\setlength{\tabcolsep}{2pt}
		\begin{tabularx}{10cm}{c|YYYYYYYYYY}
			& $v_1$ & $v_2$ & $v_3$ & $v_4$ & $v_5$ & $v_6$ & $v_7$ & $v_8$ & $v_9$ & $v_{10}$ \\
			\midrule
			\ $b_i$\ \ & \small $50$ & \small $0$ & \small $0$ & \small $0$ & \small $0$ & \small $50$ & \small $100$ & \small $100$ & \small $100$ & \small $50$ \\
		\end{tabularx}
	\end{center}
	
	Since voters $v_2$ to $v_5$ run out of money, in the third round, project B has only one supporter and project E gets 3 votes instead of 4.
	Hence, $\nicefrac{5}{6}$ of project D is bought.
	In the following rounds, the remaining $\nicefrac{1}{6}$ of project D is purchased,
	then $\nicefrac{11}{17}$ of project E,
	and finally $\nicefrac{1}{2}$ of project F
	(see \Cref{app:examples} for details).
	This leads to the following selection.
	\begin{center}
		\setlength{\tabcolsep}{2pt}
		\begin{tabularx}{6cm}{aYaaaa}
			A & B & C & D & E & F \\
			\midrule
			1 & & $\nicefrac{5}{6}$ & 1 & \nicefrac{11}{17} & \nicefrac{1}{2} \\
		\end{tabularx}
	\end{center}
	After that, the only voter with positive amount of money is $v_1$ who is left with \$50k.	
	However, the only project that $v_1$ supports has already been fully bought.
	Thus, FrES concludes.

	As a result, in our example FrES allocated \$550k to projects A--C
	and \$400k to projects D--F, while the Method of Equal Shares allocated \$300k to the former and \$410k to the latter.
	In what follows, we will propose a modification of Equal Shares
	that is more aligned with the outcomes of FrES
	and in this example spends more funds on projects A--C than on D--F.
	\hfill $\lrcorner$ 
\end{examplecont}

Since FrES may also not spend the whole budget
we can complete its outcomes in the utilitarian fashion,
i.e., 
buying projects that maximize total utility per cost until the budget is exhausted.

In the discrete model, it is established that, unless P = NP, no election rule computable in (strongly) polynomial time can satisfy EJR. 
This hardness result stems from a reduction from
the \textsc{knapsack} problem~\citep{pet-pie-sko:c:participatory-budgeting-cardinal},
and holds even for instances with a single voter.
However, this does not extend to the fractional setting as
a simple greedy algorithm can solve fractional \textsc{knapsack} optimally---in fact, FrES applied to an instance with a single voter
is equivalent to such an algorithm and
returns the optimal solution.
This opens the possibility that a polynomial-time rule, like FrES, could indeed satisfy (the fractional analog of) EJR. In what follows, we will show that this is indeed the case.

\begin{definition}[\textbf{Fractional EJR}]\label{frejr}
	A group of voters $S \subseteq V$ is $(T, \beta, \gamma)$-cohesive for $T\subseteq C$, $\beta:C\rightarrow [0,1]$, and $\gamma:C\rightarrow \mathbb{R}_{\ge 0}$ if
	\begin{enumerate}
		\item $\sum_{c\in T}\cost(c)\cdot \beta(c)  \leq b \cdot (|S| / n)$, and 
		\item $u_i(c)\cdot \beta(c)\geq \gamma(c)$ for all $c\in T$ and $v_i\in S$.
	\end{enumerate}
	A fractional outcome $W$ satisfies Fractional EJR if for every  $(T, \beta, \gamma)$-cohesive group of voters $S$ there is a voter $v_i\in S$ for which $\sum_{c\in C}u_i(c)\cdot W_c \geq \sum_{c\in T} \gamma(c)$. \hfill $\lrcorner$
\end{definition}

By fixing $\beta(c)=1$ for all $c \in C$, and by admitting only the integral solutions (that is, assuming $W_c \in \{0,1\}$),
we get the standard definition of EJR
for the integral PB model and general utilities \citep{pet-pie-sko:c:participatory-budgeting-cardinal}.
If we also fix $\gamma(c)=1$ and $u_i(c)\in \{0,1\}$ for each voter $v_i$ and candidate $c$,
we get the classic definition of EJR for approval-based committee voting~\citep{aziz2017justified}.

Additionally, if applied to approval ballots with cost utilities,
\Cref{frejr} is equivalent to Cake EJR proposed by \citet{bei2024truthful}.
Fractional Equal Shares under approval ballots with cost utilities is equivalent to
\myemph{Generalized Method of Equal Shares} introduced by \citet{lu2024approval},
who showed that their rule satisfies Cake EJR.
In the following theorem,
we generalize this result
and show that FrES satisfies Fractional EJR under arbitrary additive utilities.\footnote{Proofs of results marked by $\spadesuit$ are deferred to \Cref{app:proofs}.}

\begin{restatable}[$\spadesuit$]{theorem}{frejrthm}\label{frejr_thm}
	FrES satisfies Fractional EJR.
\end{restatable}
\toappendix{
  \sv{\frejrthm*}
\begin{proof}
	Consider an $(T,\beta,\gamma)-$cohesive set of voters $S$.
	We will prove that there exists a voter in $S$ for which the outcome of FrES results in the satisfaction of at least
	$\tilde{\gamma}=\sum_{c\in T}\min_{v_i \in S}u_i(c)\beta(c)\geq \sum_{c\in T}\gamma(c)$.
	Without loss of generality, let us assume that $\gamma(c) > 0$ for every $c \in T$.
	Our proof follows a similar strategy to the proof that MES satisfies EJR up-to-one~\citep{pet-pie-sko:c:participatory-budgeting-cardinal}.
	We examine runs of the following variants of FrES:
	\begin{itemize}
		\item[(A)] FrES in the original instance and formulated as in \cref{algorithm:fres}.
		\item[(B)] FrES where voters from $S$ can go with their money $b_i$ below zero
		when paying for projects in $T$
		(so they do not have a budget constraint when they pay for these projects,
		but they do have when paying for the rest),
		and each project $c \in T$ can be bought up to an amount of $\beta(c)$.
		\item[(C)] FrES in the instance truncated to the voters from $S$ and projects from $T$.
		The utility of a voter $v_i$ from a project $c$ is set to $u_i(c)=\min_{j\in S}u_j(c)$.
		Additionally, each voter from $S$ can go with their money $b_i$ below zero
		and each project $c \in T$ can be bought up to an amount of $\beta(c)$.
	\end{itemize}
	
	For each variant of the rule $\calR$ and each iteration $t$, we denote the value of the variables from \cref{algorithm:fres} 
	after $t$ by putting $\calR$ and $t$ as superscripts; for example, $b^{A,t}_i$ or $W^{B,t}_c$. The initial values are indicated with $t = 0$, while the final values are written without an iteration number.

	Observe that in variants (B) and (C) the utility of all voters from $S$ towards projects in $T$
	is strictly positive and no voters from $S$ have budget constraint for these projects.
	Thus, at the end of those procedures each project $c$ from $T$ will be bought up to $\beta(c)$.
	If at the end of (B) no voter $v_i$ from $S$ has negative $b_i$,
	then the solution of (B) coincides with the solution of (A).
	Thus, in such a case, for every project $c \in T$, 
	 at least $\beta(c)$ of this project will be bought in $W^A$.
	Therefore, every voter $v_i$ in~$S$ would have a total satisfaction of at least
	\begin{align*}
	\sum_{c\in T}u_i(c)\beta(c)\geq \sum_{c\in T}\min_{i\in S}u_i(c)\beta(c) =\tilde{\gamma} \text{.}
	\end{align*}
	
	Hence, in the remainder of the proof, we will consider the case
	in which a voter from $S$ overspent in (B)
	and denote the first such a voter by $v_i$.
	
	Let $f_A(x)$ be the total amount of money spent by $v_i$
	during the execution of variant (A)
	at the point when $v_i$ has exactly the utility of $x$.
	Formally, for every $x \in (0, u_i(W^A)]$,
	\[
	f_A(x) = (b/n - b^{A,t-1}_i) + \rho^{A,t} (x - u_i(W^{A,t-1})),
	\quad
	\mbox{where } t \mbox{ is such that } x \in (u_i(W^{A,t-1}), u_i(W^{A,t})].
	\]
	Here, $(b/n - b^{A,t-1}_i)$ is the money spent in the first $t-1$ rounds
	and $\rho^{A,t} (x - u_i(W^{A,t-1}))$ is the money spent in the $t$-th round
	just to the point of obtaining utility $x$.
	Let us analogously define $f_B(x)$ and $f_C(x)$.
	In what follows we will prove that
	(1) $f_C(\tilde{\gamma}) \le b/n$,
	(2) $f_B(x) \le f_C(x)$ for every $x \le \tilde{\gamma}$, and
	(3) $f_B(x) = f_A(x)$ for $x$ such that $f_B(x) = b/n$.
	Together, these three statements will imply the thesis.
	
	(1) Since in (C) the voters do not have budget constraints
	we buy all candidates from $T$,
	up to the level of $\beta$
	in increasing order of $\rho$.
	In every iteration $t \in [|T|]$, we buy candidate $c^{C,t}$ with $\rho^{C,t}$ equal to :
	\[
	\rho^{C,t} = \frac{\cost(c^{C,t})}{\sum_{v_j \in S} \min_{v_j \in S}u_j(c)} = \frac{\cost(c^{C,t})}{|S| \min_{v_j \in S}u_j(c)}.
	\]
	As we buy $\beta(c^{C,t})$ fraction of this candidate, we have
	\[
	b^{C,t}_i - b^{C,t-1}_i = \beta(c^{C,t}) \cdot \min_{v_j \in S}u_j(c) \cdot \rho = \beta(c^{C,t}) \cdot \frac{\cost(c^{C,t})}{|S|}.
	\]
	Summing this up for all $t \in [|T|]$, we have
	\[
	b^{C,|T|}_i - b^{C,0}_i =  \frac{1}{|S|} \sum_{c \in T} \beta(c) \cdot \cost(c),
	\]
	which by the first assumption of \Cref{frejr}, is not greater than $b/n$.
	Since at the end of iteration $|T|$ in (C), the utility of voter $v_i$ is exactly $\tilde{\gamma}$,
	we have that indeed, $f_C(\tilde{\gamma}) \le b/n$.
	
	(2) To show that $f_B(x) \le f_C(x)$ for every $x \le \tilde{\gamma}$,
	we will show that the left derivative of $f_B(x)$ is not greater than that of $f_C(x)$,
	for every such $x$.
	Intuitively, this means that under (B) the money of voter $v_i$ is spent 
	at least as effectively as under (C).
	Observe that the left derivative of $f_B(x)$ is just $\rho^{B,t}$
	for $t$ such that $x \in (u_i(W^{B,t-1}), u_i(W^{B,t})]$ and analogously for $f_C(x)$.
	Thus, for a contradiction, assume that there is $x \le \tilde{\gamma}, t$, and $t'$ such that
	$x \in (u_i(W^{B,t-1}), u_i(W^{B,t})]$,
	$x \in (u_i(W^{C,t'-1}), u_i(W^{C,t'})]$,
	and $\rho^{B,t} > \rho^{C,t'}$.
	Observe that for each candidate $c \in T$,
	its $\rho$ in variant (B) is not greater than in variant (C)
	as utilities for $c$ of voters in $S$ can only be larger than in (C)
	and additionally, other voters can also pay for it.
	Thus, since $\rho^{B,t} > \rho^{C,t'}$, this means that at iteration $t$ in (B)
	candidate $c^{C,t'}$ is already bought up to $\beta(c^{C,t'})$
	(otherwise we should buy candidate $c^{C,t'}$ as it has better $\rho$).
	Moreover, for every $t'' < t'$ we have that $\rho^{B,t} > \rho^{C,t'} > \rho^{C,t''}$,
	thus candidate $c^{C,t''}$ is also already bought up to $\beta(c^{C,t''})$
	at the start of the iteration $t$ in (B).
	However, this allows us to bound the utility of $v_i$ in iteration $t$ in (B) as follows
	\[
	x > \sum_{t'' \in [t']} \beta(c^{C,t''}) \cdot u_i(c^{C,t''})
	\ge \sum_{t'' \in [t']} \beta(c^{C,t''}) \cdot \min_{v_j \in S} u_j(c^{C,t''})
	\ge x,
	\]
	where the last inequality holds as $v_i$ obtains utility $x$ during iteration $t'$ in (C).
	But this is a contradiction.
	
	(3) Finally, let $t$ be the first iteration in which the money of voter $v_i$ crossed $0$,
	and let $p \in T$ be a project that was selected in this iteration.
	Observe that in all iterations up to $t$, the executions of (A) and (B) coincided.
	Moreover, unless $b^{B,t-1}_i$ is exactly zero,
	in iteration $t$ project $p$ was also selected by (A) as its $\rho$
	and $\rho$s of all other projects were the same.
	The difference is that in (A) the smaller fraction of $p$ was bought
	so as to leave $v_i$ with exactly zero money.
	Thus, $f_B(x) = f_A(x)$ for $x$ such that $f_B(x) = b/n$.
	Moreover, if such $x$ would be smaller than $\tilde{\gamma}$,
	then from the fact that $f_B$ is strictly increasing we would have that
	$f_B(x) < f_B(\tilde{\gamma}) \le f_C(\tilde{\gamma}) \le b/n$,
	which is a contradiction.
	Therefore, the utility of $v_i$ at the end of iteration $t$ in (A) is at least $\tilde{\gamma}$
	and at the end of the algorithm it cannot be smaller,
	which concludes the proof.
\end{proof}
}

% Core-stability \citep{aziz2017justified} is another, stronger, proportionality concept.
% It is natural to ask whether the solutions returned by FrES always belong to the core. Unfortunately, this is not the case. We discuss this issue in \Cref{sec:core-fractional}, where we also present an example demonstrating that the core can be empty in the fractional setting.

Aspects of fairness and proportionality
in models related to our fractional setting have also been examined by \citet{fain2016core,kroer2025computing,aziz2019fair,bogomolnaia2005collective,brandl2021distribution,munagala2022approximate,suzuki2024maximum}.

\section{Method of Equal Shares with Bounded Overspending}
\label{sec:bos}

In this section we build upon the idea behind Fractional Equal Shares, and design a new method for the standard (integral) model of participatory budgeting.
The Method of Equal Shares with Bounded Overspending, in short, BOS Equal Shares, or just BOS, can be viewed as a rounding procedure for FrES, but also as a variant of Equal Shares, where voters may occasionally spend more than their initial entitlement.

Under the fractional rule, an $\alpha$-fraction of a candidate can be purchased for a corresponding fraction of its cost; then each voter receives a proportional fraction of the utility. 
However, in the integral model, fractional purchases are not possible. 
In BOS, we simulate buying an $\alpha$-fraction of a candidate, still assuming the voters gain fractions of utilities. However, the voters are now required to cover the full cost.
The distribution of the payments is proportional to the one computed for the $\alpha$-fraction.
Specifically, for $\rho$ satisfying
\begin{align*}
	\alpha \cdot \cost(c) = \sum_{v_i \in V} \min(b_i, \alpha \cdot u_i(c) \cdot \rho),
\end{align*}
a voter $v_i$ supporting $c$ has to pay $p_i(c) = \min(b_i, \alpha \cdot u_i(c) \cdot \rho)/\alpha$.
Notice that we divide by $\alpha$ as we now need to cover the entire cost of the project,
not only its $\alpha$ fraction.
Assuming $\alpha \cdot u_i(c)$ is the utility from candidate $c$, the highest payment per unit of utility is then $\rho / \alpha$.
Observe that $p_i(c)$ can be actually greater than $b_i,$ i.e., the remaining budget of voter $v_i$.
In such a case we say that voter $v_i$ is \myemph{overspending}
and we set its account to zero.

\begin{description}
	\item[BOS Equal Shares] 
	Let $b_i$ be the virtual budget of voter $v_i$, initially set to $b_i := \nicefrac{b}{n}$. 
	Each round, among all candidates that fit within the remaining budget, BOS buys the $(\alpha, \rho)$-affordable candidate $c$ with the lowest possible value of $\rho/\alpha$. 
	The accounts of the supporters of $c$ are updated accordingly: $b_i := \max(0, b_i - u_i(c) \cdot \rho)$. 
	The method stops if no remaining candidate fitting within the budget is $(\alpha, \rho)$-affordable for some $\alpha \in (0,1]$ and $\rho \in \mathbb{R}_+$.
\end{description}
Note that voters without money, including these that overspent before, do not have an impact on the further decisions.

\SetKw{Return}{return}
\SetKw{Input}{Input:}
\begin{algorithm}[t]
	\small
	\captionsetup{labelfont={sc,bf}, labelsep=newline}
	\caption{Pseudo-code of the Method of Equal Shares with Bounded Overspending.}\label{alg:bounded-overspending}
	\Input{A PB election $(C,V,b)$\\}
	\DontPrintSemicolon
	\SetAlgoNoEnd
	\SetAlgoLined
	$W \gets \emptyset$ \; 
	$b_i \gets \nicefrac{b}{n}  \text{~for each~} v_i \in V$ \;
	\While{$C' = \{ c \in C \setminus W : \cost(c) \le b-\cost(W) \mbox{ \emph{and} } \sum_{v_i \in V: u_i(c) > 0} b_i > 0\}$ \emph{is nonempty}}{
		$(\alpha^*, \rho^*, c^*) \gets (1, +\infty, c)$ \;
		
		\For{$c \in C'$}{
			$\lambda' \gets \lambda$ satisfying $\cost(c) = \sum_{i=1}^n \min(b_i, u_i(c) \cdot \lambda)$ or $+\infty$ if there is no such $\lambda$, i.e., $\cost(c)>\sum_{v_i \in V : u_i(c) > 0}b_i$\;
			\For{$\lambda \in \{b_i/u_i(c) : v_i \in V, b_i >0, u_i(c) > 0\} \cup \{\lambda'\}$}{
				$\alpha \gets \min(\left(\sum_{i=1}^n \min(b_{i},u_{i}(c) \cdot \lambda)\right)/\cost(c),1)$ \;
				$\rho \gets \lambda / \alpha$ \;
				\If{$\rho/\alpha < \rho^*/\alpha^*$}{
					$(\alpha^*, \rho^*, c^*) \gets (\alpha, \rho, c)$ \;
				}
			}
		}
		$W \gets W \cup \{c^*\}$ \;
		\For{$v_i \in V$ \emph{such that} $b_i > 0$ \emph{and} $u_i(c^*) > 0$}{
			$b_i \gets \max(0, b_i - u_i(c^*) \cdot \rho^*)$ \;	
		}
	}
	\Return{W}\;
\end{algorithm}

We provide a pseudo-code of the Method of Equal Shares with Bounded Overspending in \Cref{alg:bounded-overspending} and 
we begin the discussion by analyzing how BOS works on our running example.

\begin{examplecont}
	Consider the instance depicted in \Cref{tab:running_example}.
	In the first round, BOS, as the Method of Equal Shares, selects project A,
	which is the one with the highest number of votes,
	and equally distributes its cost among the supporters
	(see \Cref{app:examples} for details).
	After the purchase, voters have the following amount of money.
	\begin{center}
		\setlength{\tabcolsep}{2pt}
		\begin{tabularx}{10cm}{l|YYYYYYYYYY}
			& $v_1$ & $v_2$ & $v_3$ & $v_4$ & $v_5$ & $v_6$ & $v_7$ & $v_8$ & $v_9$ & $v_{10}$ \\
			\midrule
			\ $b_i$\ \  & \small $50$ & \small $50$ & \small $50$ & \small $50$ & \small $50$ & \small $50$ & \small $100$ & \small $100$ & \small $100$ & \small $100$ 
		\end{tabularx}
	\end{center}			
	
	In the second round the remaining budget is $\$700\text{k}$. 	
	The following table shows values of $\rho$ for which each project is $(\alpha, \rho)$-affordable
	for each $\alpha \in \{1, \nicefrac{5}{6}, \nicefrac{5}{8}\}$.
	\begin{center}
		\setlength{\tabcolsep}{2pt}
		\begin{tabularx}{7cm}{l|aYaYYY}
			& A & B & C & D & E & F \\
			\midrule
			\ $\alpha = 1, \rho$\ \  & -- & -- & \nicefrac{1}{3} & \nicefrac{1}{4} & \nicefrac{1}{4} & \nicefrac{1}{3} \\
			\ $\alpha = \nicefrac{5}{6}, \rho$\ \  & -- & -- & \nicefrac{1}{5} & \nicefrac{1}{4} & \nicefrac{1}{4} & \nicefrac{1}{3} \\
			\ $\alpha = \nicefrac{5}{8}, \rho$\ \ & -- & \nicefrac{1}{5} & \nicefrac{1}{5} & \nicefrac{1}{4} & \nicefrac{1}{4} & \nicefrac{1}{3} \\
		\end{tabularx}
	\end{center}
	Since voters $v_1$ to $v_6$ have only $\$50\text{k}$ each, project B cannot be purchased in full.
	However, its supporters can cover $\alpha = \nicefrac{250}{400} = \nicefrac{5}{8}$ of its cost.
	This cost would be equally spread among five voters, so $\rho = \nicefrac{1}{5}$ and the ratio ($\rho$ scaled by $\alpha$) equals $\rho/\alpha = \nicefrac{8}{25}$.
	In principle, we could buy this project also with $\alpha < \nicefrac{5}{8}$,
	but since we have equal payments, $\rho$ would still be equal to $\nicefrac{1}{5}$,
	hence overall $\rho/\alpha$ would decrease in such a case.
	
	Now, let us consider project C.
	This project can be bought in full, but only if voter $v_{10}$ pays $\rho = \nicefrac{1}{3}$ of the cost.
	The Method of Equal Shares rejects this option as imbalanced and selects a project with a smaller $\rho$ parameter.
	In turn, BOS considers also buying a fraction of the project with balanced payments.
	Specifically, to maintain equal payments, $\alpha=\nicefrac{5}{6}$ of the project can be bought.
	This would result in $\rho=\nicefrac{1}{5}$ and ratio $\rho/\alpha = \nicefrac{6}{25}$.
	Since this ratio is smaller than the ratio of project B and the ratio of projects D, E and F which remain unchanged,
	BOS Equal Shares selects project C with $\alpha = \nicefrac{5}{6}$.
	Again, in principle, we could consider different values of $\alpha$,
	but this is unnecessary, which we will soon formally prove (see \Cref{theorem:bos_thm}).
	To cover $\nicefrac{5}{6}$ of the cost of project C, its supporters would have to pay $\$50\text{k}$ each.
	Hence, in BOS Equal Shares, each voter pays $\nicefrac{6}{5} \cdot \$50\text{k} = \$60\text{k}$.
	This means that some voters overspend their available funds,
	which effectively means that we set their accounts to 0.
	\begin{center}
		\setlength{\tabcolsep}{2pt}
		\begin{tabularx}{10cm}{l|YYYYYYYYYY}
			& $v_1$ & $v_2$ & $v_3$ & $v_4$ & $v_5$ & $v_6$ & $v_7$ & $v_8$ & $v_9$ & $v_{10}$ \\
			\midrule
			\ $b_i$\ \  & \small $50$ & \small $0$ \tiny$(-10)$ & \small $0$ \tiny$(-10)$ & \small $0$ \tiny$(-10)$ & \small $0$ \tiny$(-10)$ & \small $50$ & \small $100$ & \small $100$ & \small $100$ & \small $40$
		\end{tabularx}
	\end{center}			
	
	In the third round BOS selects project D with $\alpha = 1$ and $\rho = \nicefrac{5}{18}$
	and in the fourth round project F is bought with $\alpha = \nicefrac{5}{6}$ and $\rho = \nicefrac{3}{5}$
	(see \Cref{app:examples} for details).
	After that, the remaining budget is \$60k, no project can be added, and BOS terminates.

	The final outcome is then $\{A, C, D, F\}$.
	Note that similarly to Fractional Equal Shares and in contrast to the original Method of Equal Shares,
	BOS spends more funds on projects A--C than on D--F,
	which is reasonable as the former projects are more popular.
	\hfill $\lrcorner$
\end{examplecont}

As we have noted in the above example, we do not have to consider all values of $\alpha$. 
In fact, it can be shown that it is optimal to either consider $\alpha = 1$ or a fraction $\alpha$ that, for some voter, is the smallest amount for which she runs out of money, assuming proportional payments, i.e., in which $b_i = u_i(c) \cdot \rho$ for some voter $v_i \in V$. 
In such a case, it holds that $\rho = b_i/u_i(c)$. 
Since there are at most $n$ such values, the outcome of BOS can be computed in polynomial time.

\begin{restatable}{theorem}{bosthm}
	\label{theorem:bos_thm}
	BOS runs in polynomial time.
\end{restatable}

\begin{proof}
	It is enough to argue that in each round, the next project can be selected in polynomial time, or equivalently, that we can find parameters $\rho$ and $\alpha$ that minimize $\rho/\alpha$ in polynomial time.
	
	Consider a fixed round of BOS, and assume each voter has a budget $b_i$. Let $S$ be the set of voters who still have money left.
	Fix a not-yet-elected candidate $c$ that has supporters in $S$ (otherwise, it is not $(\alpha,\rho)$-affordable for any finite $\rho$).
	If $c$ can be fully funded by voters from $S$ with payments proportional to their utilities, then clearly, the optimal $\alpha$ equals $1$, and $\rho = \cost(c) / \sum_{v_i \in S} u_i(c)$.
	
	Assume otherwise.
	Consider a sequence $(v_{i_1}, \dots, v_{i_{\ell}})$ of supporters of $c$ who still have money left, sorted in the ascending order by $b_i/u_i(c)$. 
	Note that if we gradually increase the fraction $\alpha$ of the project that we fund, and for each such value of $\alpha$ minimize the value of $\rho$, then the first voter to run out of money would be $v_{i_1}$, the second (or possibly ex-aequo first) would be $v_{i_2}$, and so on. 
	Let $\alpha_j$ be the fraction of the project for which voter $v_{i_j}$ runs out of money (or $\alpha_j=1$ if they do not).
	Assume $\alpha_0 = 0$ and consider buying a fraction $\alpha \in [\alpha_j, \alpha_{j+1}]$ of the project for some $j \in \{0,\dots,\ell-1\}$.
	We argue that $\alpha$ equal to $\alpha_j$ or $\alpha_{j+1}$ minimizes the ratio $\rho/\alpha$ on this interval.
	
	For $\alpha \in [\alpha_j, \alpha_{j+1}]$, we know that each voter $v_i \in \{v_{i_1},\dots,v_{i_{j}}\}$ would pay $b_i$, and each voter $v_i \in \{v_{i_{j+1}},\dots,v_{i_{\ell}}\}$ would pay $\alpha \cdot u_i(c) \cdot \rho$.
	Hence, we have $\alpha \cdot \cost(c) = \sum_{h=1}^j b_{i_h} + \sum_{h={j+1}}^{\ell} \alpha \cdot u_{i_h}(c) \cdot \rho,$ and, therefore
\[ \rho/\alpha = \frac{\alpha \cdot \cost(c) - \sum_{h=1}^j b_{i_h}}{\alpha^2 \cdot \sum_{h={j+1}}^{\ell} u_{i_h}(c)}. \] 
	Taking the derivative with respect to $\alpha$, we get that 
	\[ \frac{d}{d \alpha}(\rho/\alpha) = \frac{2 \sum_{h=1}^j b_{i_h} - \alpha \cost(c)}{\alpha^3 \sum_{h=j+1}^{\ell} u_{i_h}(c)}. \]
	Since $\cost(c) > 0$, this function has only one zero at $\alpha^* = 2 \sum_{h=1}^j b_{i_h} / \cost(c)$, and it is positive for $\alpha < \alpha^*$ and negative for $\alpha > \alpha^*$. 
	This implies that $\rho/\alpha$ is increasing for $\alpha < \alpha^*$ and decreasing for $\alpha > \alpha^*$.
	Hence, for any interval, the minimum value is at one of its ends, and so $\rho/\alpha$ is minimal for $\alpha \in \{\alpha_j, \alpha_{j+1}\}$.
	Note that for the interval $[\alpha_0, \alpha_1]$ all payments are proportional to their utilities, hence $\rho$ is a constant and $\rho/\alpha$ is maximized for $\alpha = \alpha_1$.
	
	This shows that the fraction $\alpha$ that minimizes $\rho/\alpha$ belongs to the set $\{\alpha_1,\dots,\alpha_{\ell}\} \cup \{1\}$.
	If $\alpha = \alpha_j < 1$ for some $i_j \in \{1,\dots,\ell\}$, then we get $\rho = b_{i_j}/u_{i_j}(c)$.
	On the other hand, if $\alpha = 1$, then $\rho$ can be computed as in the Method of Equal Shares.
\end{proof}

\subsection{Addressing the Limitations of the Method of Equal Shares}

Next, let us discuss how BOS addresses each of the limitations of Equal Shares
that we have identified and listed in the introduction of our work.

For underspending,
observe that BOS Equal Shares leads to a non-exhaustive outcome
only if all voters that still have remaining funds have all of their supported projects selected
(otherwise there is an $(\alpha,\rho)$-affordable project for some, possibly very small, $\alpha$).
Based on our empirical analysis (see \Cref{sec:experiments})
such a situation is extremely rare in practice.
As a result, BOS spends on average a similar fraction of a budget as the Utilitarian Method
and does not need to be combined with a completion mechanism,
in contrast to the Method of Equal Shares.
However, if an exhaustive rule is required,
a simple modification can be made to BOS Equal Shares:
whenever all projects supported by a voter are already selected,
we remove this voter from the election
and redistribute her remaining funds equally among the remaining voters.
We note that similar modification for the original Method of Equal Shares
would still lead to a non-exhaustive rule.

Next, we discuss how our method handles the instances of PB elections presented in \Cref{sec:challenging-instances}.

\stepcounter{exof}
\begin{exof}[BOS on Helenka Paradox]
	\label{ex:bos_helenka}

	Consider the Helenka Paradox election instance presented in \Cref{sec:challenging-instances}.
	The initial endowment is $\$(\nicefrac{310,000}{414}) \approx \$750$.
	Only the $\alpha = \nicefrac{403}{414}$ fraction of the project A can be paid for by its supporters.
	This gives $\rho =\nicefrac{1}{403}$, which results in $\rho/\alpha = \nicefrac{414}{403^2}$.
	In turn, project B can be fully bought, thus we have $\alpha = 1$ and $\rho = \nicefrac{1}{11} = \rho/\alpha$.
	Hence, BOS, in contrast to the Method of Equal Shares, would select project A.
	
	The reader may wonder what would happen if the support for project B was higher.
	Clearly, as the support for project B that is affordable at full increases
	and that of project A taking the whole budget decreases,
	at some point we should switch from selecting A to selecting B.
	At what point would BOS start selecting project B?
	Assume that $x$ voters support A and $n-x$ voters support B.
	Then, for A we have $\rho= 1/x$ and $\alpha=x/n$, thus $\rho/\alpha = n/x^2$.
	In turn, for B we have $\rho= 1/(n-x)$ and $\alpha=1$, thus $\rho/\alpha = 1/(n-x)$.
	Hence, we can select project B if $1/(n-x) \le n/x^2$ or equivalently $n^2 - xn \ge x^2$.
	Solving the quadratic equation, we get that if $x>0$,
	then this is equivalent to $x/n \le (\!\sqrt{5} -1)/2 = 1/\varphi$.
	Therefore, it turns out that the switching point is when the supports of the projects
	are in the golden ratio to each other, i.e.,
	the project using the entire budget is supported by roughly 61.8\% of the electorate.
	This is close to the value of the thresholds required for a supermajority in many democratic systems (e.g., 3/5 or 2/3).
	\hfill $\lrcorner$
\end{exof}

\begin{exof}[BOS on Tail Utilities]
	\label{ex:bos_tail} 
	Consider the Tail Utilities instance from \Cref{sec:challenging-instances}.
	All voters have the same utility from project B, hence, every fraction $\alpha$ of the project can be bought with equal payments and $\rho = \cost(c)/\sum_{v_i \in V} u_i(c) = \nicefrac{1}{200}$.
	Therefore, we obtain the minimal ratio $\rho/\alpha = \nicefrac{1}{200},$ for $\alpha = 1$.
	
	Let us focus now on project A. 
	If the project is bought in full, the last voter would have to pay $\$0.01$ per unit of utility, giving $\rho = \nicefrac{1}{100}$.
	However, it is also possible to simulate the purchase of the fraction of the project that ensures the payments are proportional to utilities.
	This happens if 99 voters pay $\$0.01$ and 1 voter pays $\$0.0001$.
	This way, voters pay for a fraction of $\alpha = \nicefrac{9901}{10000}$ of the project, achieving a much better ratio: $\rho = \cost(c)/\sum_{v_i \in V} u_i(c) = \nicefrac{1}{9901}$.
	Thus, we get $\rho/\alpha = \nicefrac{10000}{9901^2}$, and so BOS would select project A.
	\hfill $\lrcorner$
\end{exof}

\subsection{Proportionality Guarantees of BOS Equal Shares}

As previously noted, Extended Justified Representation (EJR) is a well-established axiom of proportionality for participatory budgeting. In \Cref{ex:bos_helenka}, we illustrate that BOS fails to satisfy EJR under cost utilities in the PB setting. Importantly, this failure is not a deficiency but an intended behavior. Nevertheless, BOS still provides strong proportionality guarantees. We primarily demonstrate this through extensive experiments on real data (see \Cref{sec:experiments}), and we additionally prove theoretically that the extent of violation of EJR is limited.

In line with the definition of EJR under cost utilities, given a subset of projects $T \subseteq C$, we say that a group of voters $S$ \myemph{deserves} the satisfaction of $\cost(T)$ if $|S|\geq \cost(T)\frac{n}{b}$ and all voters in $S$ approve all projects in $T$, that is $T \subseteq \bigcap_{i \in S} A_i$.

\begin{definition}
	\label{ejrPB}
	Assuming cost utilities, we say that a rule satisfies EJR up to $t$, if for every election $E = (C,V,b)$, and $W$ being the outcome of the rule on $E$, the following holds:
	For each pair $(S,T)\subseteq V\times C$ such that $S$ deserves the satisfaction of $\cost(T)$, there is a voter $v_i$ such that $u_i(W) \geq \cost(T)-t-\cost(c),$ for all $c\in T\setminus W$.
	\hfill $\lrcorner$
\end{definition}

\noindent The following theorem is the main theoretical contribution of the section.

\begin{theorem}
	\label{thm:posPB}
	For cost utilities, BOS satisfies EJR up to $\frac{n-|S|}{2|S|}\cdot c^*,$ where $c^*:=\max_{c\in C}\cost(c).$
\end{theorem}

\begin{proof}
	Consider a subset of voters $S \subseteq V$ and a subset of candidates $T \subseteq C$ such that $S$ deserves the satisfaction of $\cost(T)$.
	Towards a contradiction, assume that there exists a candidate $\hat{c}\in T\setminus W$ such that
	for all voters $v_i\in S$ it holds that $u_i(W) < \cost(T)-\frac{n-|S|}{2|S|}\cdot c^*-\cost(\hat{c})$.
	
	\begin{claim}\label{1/s}
		If there is a candidate $\hat{c} \in T\setminus W$ such that for all $v_i\in S$ it holds that $u_i(W) < \cost(T)-\frac{n-|S|}{2|S|}\cdot c^*-\cost(\hat{c})$,
		then, in every step of the execution of BOS, it holds $\frac{\rho}{\alpha} \leq \frac{1}{|S|}$.
	\end{claim}
	\begin{proof}
		Towards a contradiction, suppose that in some round $r$ the rule selected a candidate $c$ with $\nicefrac{\rho}{\alpha}>\nicefrac{1}{|S|}$. Assume that $r$ was the first round when this happened.
		Thus, all previous candidates were bought with \(\nicefrac{\rho}{\alpha} \leq \nicefrac{1}{|S|}\). 
		Consider a candidate $\hat{c} \in T \setminus W$. At the time $c$ was bought, there was a voter from $S$, say $v_i$ who had spent strictly more than $\nicefrac{b}{n}-\nicefrac{\cost(\hat{c})}{|S|}$, as otherwise the voters from $S$ would altogether have at least $|S| \cdot (\cost(\hat{c}) / |S|)$ money left, and so $\hat{c}$ could have been bought for $\alpha=1$ and $\rho=\nicefrac{1}{|S|},$ i.e., for $\nicefrac{\rho}{\alpha}= \nicefrac{1}{|S|}$. We now compute the satisfaction of voter $v_i$ with respect to $W$. By the fact that any purchase prior to $c$ was done for $\rho\leq \nicefrac{1}{|S|},$ which is due to the fact that $\alpha$ is upper bounded by $1$, we have that the satisfaction of $v_i$ whenever paying $p$ was increasing by at least $|S|\cdot p$. Therefore, the satisfaction of $v_i$ at the considered step was 
		\begin{align*}
		u_i(W) \ge |S|\left(\frac{b}{n}-\frac{\cost(\hat{c})}{|S|}\right)\geq \cost(T)-\cost(\hat{c}) \text{,}
		\end{align*}
		which contradicts the fact that there are no such voters.
	\end{proof}

	\begin{claim}\label{half_overspend_pb}
		Consider a round when some voters overspend.
				The number of voters that exhaust their budgets in this round is strictly larger 
		than half of the number of all voters who paid for a candidate in this round.
	\end{claim}
	\begin{proof}
		Consider a round of the BOS procedure in which a candidate $c$ is selected
		with certain $\alpha$ and $\rho$.
		Let $R$ be the set of all voters that paid for $c$.
		The voters from $R$ can be partitioned into two subsets:
		those who either overspent their budget or spent all of their available funds on $c$,
		i.e., they exhausted their budgets in this round,
		and the remaining voters---those who will still have money to spend in future rounds.
		Let us denote the first subset by $R_{\supminus}$ and the second by $R_{\supplus}$.
		Clearly $R = R_{\supminus} \cup R_{\supplus}$ and both subsets are disjoint.
		In what follows, we will show that $|R_{\supminus}| > |R|/2$.

		If $R_{\supplus} = \varnothing$, the thesis holds trivially,
		thus let us assume that $|R_{\supplus}|>0$.
		Observe that every voter in $R_{\supplus}$ pays exactly the same amount in this round
		(as we assumed cost utilities and their budgets are not finished).
		Let us denote the amount each of them pays for an $\alpha$ fraction of $c$ by $p$,
		i.e., they pay ${p}/{\alpha}$ altogether.
		Since the voters in $R_{\supplus}$ pay the highest price,
		we get $\rho = \frac{p}{\alpha \costc}$ and $\frac{\rho}{\alpha} = \frac{p}{\alpha^2 \costc}$.

		Now, consider buying candidate $c$ with alternative $\bar{\alpha}$ and $\bar{\rho}$
		which we obtain by making voters in $R_{\supplus}$ pay $p+\beta$ instead of $p$ for a fraction of $c$,
		where $\beta>0$ is small enough that each $v_i \in R_{\supplus}$ has at least $(p + \beta)/\alpha$ money
		(since after paying ${p}/{\alpha}$ they still had a positive amount of money, there is such $\beta$).
		This means that $\bar{\alpha} \ge \alpha + \beta \cdot |R_{\supplus}| / \costc$,
		as we account for the additional payments.
		Also, $\bar{\rho} = \frac{p + \beta}{\bar{\alpha} \costc}$.
		Since BOS has chosen to buy project $c$ with $\alpha$ and $\rho$, we get
		\[
			\frac{p}{\alpha^2 \costc} = 
			\frac{\rho}{\alpha} \le \frac{\bar{\rho}}{\bar{\alpha}} = 
			\frac{p + \beta}{\bar{\alpha}^2 \costc} \le \frac{p + \beta}{(\alpha + \beta \cdot |R_{\supplus}|/\costc)^2 \costc}.
		\]
		After rearrangement we obtain
		\[
			(\alpha \costc + \beta|R_{\supplus}|)^2 \le \alpha^2 \costc^2 \frac{(p + \beta)}{p} = \alpha^2 \costc^2 ( 1 + \beta/p).
		\]
		Observe that the equal payments in $R$ would result in
		every voter paying $\alpha\costc / |R|$ for the $\alpha$ fraction of $c$.
		As voters in $R_{\supplus}$ pay the largest share, we get $p \ge \alpha\costc / |R|$.
		Thus, we have that $1 + \frac{\beta}{p} \le \frac{|R|\beta + \alpha\costc}{\alpha\costc}$.
		Therefore,
		\[
			\alpha^2 \costc^2 + 2\alpha \costc\beta|R_{\supplus}| + \beta^2|R_{\supplus}|^2 =
			 (\alpha \costc + \beta|R_{\supplus}|)^2 \le 
			 \alpha \costc(|R|\beta + \alpha\costc).
		\]
		As $\alpha\costc > 0$ and $\beta > 0$, we divide by $\alpha\costc$, subtract $\alpha\costc$, and divide by $\beta$,
		which yields
		\[
			2|R_{\supplus}| + \frac{\beta|R_{\supplus}|^2}{\alpha\costc} \le |R| \Leftrightarrow
%		\]
%		Hence,
%		\[
			\frac{\beta|R_{\supplus}|^2}{\alpha\costc} \le |R| - 2|R_{\supplus}| = 2|R_{\supminus}| - |R|.
		\]
		Since the left-hand side is strictly positive, we get that $|R_{\supminus}| > |R|/2$,
		which concludes the proof.
	\end{proof}

	For every voter $v_i \in V$, let us denote the overspending of $v_i$ by
	$\delta_i = -\min (0, b_i - u_i(c^*) \cdot \rho^*)$,
	where $c^*$ is the candidate chosen in the last round in which $v_i$ was paying,
	$\rho^*$ is $\rho$ in that round, and $b_i$ is the remaining budget of $v_i$ at the beginning of that round.
	\begin{claim}\label{overspendPB}
		If there is a candidate $\hat{c} \in T\setminus W$ such that for all $v_i\in S$ it holds that $u_i(W) < \cost(T)-\frac{n-|S|}{2|S|}\cdot c^*-\cost(\hat{c})$, then, the total overspending is bounded by $\sum_{v_j \in V}\delta_j \le \frac{n-|S|}{2|S|}\cdot c^*$.
	\end{claim}
	\begin{proof}
		First, observe that the voters from $S$ will never overspend, i.e., $\delta_i=0$, for every $v_i \in S$. 
		Indeed, from \Cref{1/s} we know that $\rho \le 1/|S|$ in every round,
		which combined with the fact that $u_i(W) < \cost(T)$,
		bounds the total amount of money that voter $v_i$ spends by
		\begin{align*}
			u_i(W) \rho < \cost(T) \rho \leq \cost(T) \cdot (1/|S|) \leq b/n .
		\end{align*}

        Now, consider a subset of voters, $R \subseteq V$, that pay for candidate $c$
		selected with $\alpha$ and $\rho$ in a certain round of BOS.
		In what follows, we will bound the overspending in this round.
		From the proof of \Cref{half_overspend_pb} and using the notation from there
		we know that 
		\(
			\frac{\rho}{\alpha} = \frac{p}{\alpha^2\costc}
		\)
		and
		\(
			p \ge \frac{\alpha\costc}{|R|}.
		\)
		Then, \Cref{1/s} yields
		\(
			|S| \le \frac{\alpha}{\rho} = \frac{\alpha^2\costc}{p} \le \alpha|R|.
		\)
		Hence, $\alpha \ge |S|/|R|$.
		On the other hand, observe that the overspending of voters in $R_{\supminus}$
		(so all voters that exhaust their funds in this round)
		cannot be greater than what remains after buying $\alpha$ fraction of $c$, i.e.,
		\[
			\sum_{v_i \in R_{\supminus}} \delta_i \le (1 - \alpha)\costc \le c^* \big(1 - |S|/|R|\big).
		\]
		Then, by \Cref{half_overspend_pb}, the average overspending of voters in $R_{\supminus}$
		is at most
		\[
			\sum_{v_i \in R_{\supminus}} \frac{\delta_i}{|R_{\supminus}|} \le
			\frac{c^* (|R| - |S|)/|R|}{|R|/2} =
			\frac{2c^*(|R| - |S|)}{|R|^2} \le 
			\frac{c^*}{2|S|},
		\]
		where the last inequality follows from the fact that
		\(
			|R|^2 - 4|S|\!\cdot\!|R| + 4|S|^2 = (|R| - 2|S|)^2 \ge 0.
		\)
		Thus, the average overspending in $V \setminus S$ is also bounded by $c^*/(2|S|)$,
		which concludes the proof.
	\end{proof}

	Recall that $\hat{c}\in T\setminus W$
	is such that $u_i(W) < \cost(T)-\frac{n-|S|}{2|S|}\cdot c^*-\cost(\hat{c})$,
	for every $v_i \in S$.
	Observe that the voters in $V \setminus S$ can spend in total at most 
	$(n-|S|) \cdot b /n$ from their initial budgets
	and, by \Cref{overspendPB}, $\frac{n - |S|}{2 |S|}c^*$ from overspending.
	Thus, voters from $S$ must have spent at least
	\[
	b - \left(\frac{n-|S|}{n}b + \frac{n - |S|}{2 |S|}c^*\right) - \cost(\hat{c})  = \frac{|S|}{n}b - \frac{n - |S|}{2 |S|}c^* - \cost(\hat{c})  \geq \cost(T) - \frac{(n - |S|) c^*}{2 |S|} - \cost(\hat{c})
	\] 
	as otherwise they could afford to buy $\hat{c}$.

	From \Cref{1/s}, we know that any payment of $p$ for a voter in $S$ results in a satisfaction of at least $|S| \cdot p$.
	Thus, there must exist a voter in $S$ with a satisfaction of at least
	\begin{align*}
		\frac{1}{|S|} \cdot \left(\cost(T)-\frac{n-|S|}{2|S|}c^*-\cost(\hat{c})\right) \cdot \frac{1}{\rho} \geq \cost(T)-\frac{n-|S|}{2|S|}c^*-\cost(\hat{c}).
	\end{align*}
	This gives a contradiction and concludes the proof of the theorem.
\end{proof}

\begin{corollary}
	\label{cor:positiveEJR-bos}
	For approval-based committee elections, BOS satisfies EJR up to $\left\lceil\frac{k-\ell}{2\ell}\right\rceil$ candidates,
	where $k$ is the size of the committee, and $\ell = |S| \cdot k / n$ is the number of representatives that voters in $S$ deserve.
\end{corollary}

To illustrate this result, consider a group entitled to $\sqrt{k}$ representatives.
BOS ensures that at least one voter in the group receives no fewer than $\nicefrac{\sqrt{k}}{2}$ candidates she likes.
Furthermore, as the group size increases, this guarantee becomes even stronger.

Our guarantees are asymptotically tight.

\begin{restatable}[$\spadesuit$]{proposition}{negEJRbos}\label{prop:negEJRbos}
	For each $\ell$ there exists an approval-based committee election where a group of voters deserves $\left\lceil\frac{k-\ell}{4\ell}\right\rceil$ candidates, and all voters from this group get no representatives under BOS. 
\end{restatable}
\toappendix{
  \sv{\negEJRbos*}
  \begin{proof}
	Fix $\ell\geq 1$ and consider the following instance of approval-based committee elections with $|V| = 4\ell^2 + \ell$ voters and $|C| = 4\ell^2 + 2\ell$ candidates; the available budget is equal to $b = 4\ell^2 + \ell$. The preferences of the voters are as follows:
	\begin{itemize}
		\item The first $\ell$ voters, referred to as group $V_1$, approve all $\ell$ candidates from $C_1 = \{c_1, c_2, \ldots, c_{\ell}\}$.
		\item The remaining $4\ell^2$ voters, referred to as group $V_2$, all approve $4\ell^2 - \ell$ common candidates from $C_2 = \{c_{\ell + 1}, c_{\ell + 2}, \ldots, c_{4\ell^2}\}$. These $4\ell^2$ voters are further divided into $2\ell$ subgroups, $V_2^{(1)}, V_2^{(2)}, \ldots, V_2^{(2\ell)}$, each consisting of $2\ell$ voters who each approve one of the remaining $2\ell$ candidates. Let us denote the set of these candidates as $C_3 = \{c_{4\ell^2 + 1}, c_{4\ell^2 + 2}, \ldots, c_{4\ell^2 + 2\ell}\}$.
	\end{itemize}
	According to the EJR, the voters in $V_1$ are entitled to $\ell$ candidates, since they all approve the same set of $\ell$ candidates and $|V_1| \geq \ell \cdot \nicefrac{|V|}{b} = \ell$. Further, $\ell = \left\lceil\frac{b-\ell}{4\ell}\right\rceil$.
	However, we will demonstrate that the outcome of BOS will not include any candidates from $C_1$.
	
	Each voter starts with a unit budget. A candidate from $C_1$ can be purchased with $\rho = \nicefrac{1}{\ell}$ and $\alpha = 1$. Similarly, a candidate from $C_2$ can be bought with  $\rho = \nicefrac{1}{4\ell^2}$ and $\alpha = 1$, and a candidate from $C_3$ with $\rho = \nicefrac{1}{2\ell}$ and $\alpha = 1$. Selecting all $4\ell^2 - \ell$ candidates from $C_2$ leaves each voter in $V_2$ with a remaining budget of:
	\begin{align*}
	1 - (4\ell^2 - \ell) \cdot \frac{1}{4\ell^2} = \frac{1}{4\ell} \text{.}
	\end{align*}
	At this stage, a candidate from $C_3$ can be purchased with $\rho = \nicefrac{1}{2\ell}$ and $\alpha = \nicefrac{1}{2}$, which results in $\nicefrac{\rho}{\alpha} = \nicefrac{1}{\ell}$. Therefore, BOS may choose any candidate from $C_1$ or $C_3$. If we break ties in favor of $C_3$, then ultimately, BOS will select all candidates from $C_2$ and $C_3$.
\end{proof}
}

\subsection{Method of Equal Shares with Bounded Overspending Plus}

While BOS solves multiple problems of Equal Shares, one can still argue that it is not always evident that it produces the most desirable outcome.

\begin{exof}[Selection of Unpopular Projects]
	\label{ex:unpopular projects}
	Consider an election with $m=310$ unit-cost candidates, and $n = 1000$ voters. 
	The budget is $b=\$10$. Voters $v_1,\dots,v_{700}$ approve ten projects $\text{A}_1$ to $\text{A}_{10}$, and each voter from $v_{701}$ to $v_{1000}$ approves a single project from $\text{B}_1$ to $\text{B}_{300}$; each such project is approved by only one voter. See \Cref{table:ex:unpopular projects}.
	
	\begin{table}[t]
		\centering
		\begin{tabularx}{\columnwidth}{lc|YYYYYY}
			& cost & $v_1$ & $\dots$ & $v_{700}$ & $v_{701}$ & $\dots$ & $v_{1000}$ \\
			\midrule
			Project $\text{A}_1$ & \$1 & \faCheck & \faCheck & \faCheck & & & \\
			\ldots & \$1 & \faCheck & \faCheck & \faCheck & & & \\
			Project $\text{A}_{10}$ & \$1 & \faCheck & \faCheck & \faCheck & & & \\
			Project $\text{B}_1$ & \$1 & & & & \faCheck & & \\
			\ldots & \$1 & & & & & \faCheck & \\
			Project $\text{B}_{300}$ & \$1 & & & & & & \faCheck \\
		\end{tabularx}
	\caption{An example of PB election in which BOS may output suboptimal outcome.}
	\label{table:ex:unpopular projects}
	\end{table}
	
	Consider how BOS operates on this instance.
	First, it would buy seven A-projects. 
	After that, voters $v_1$ to $v_{700}$ would not have any money left, so BOS would additionally buy three B-projects,
	which is arguably not the most effective allocation of the funds.
	On the other hand, the Method of Equal Shares would buy only seven A-projects, but its Add1U variant would select all ten A-projects, i.e., it would return the outcome for the initial endowment equal to $b_{\ini} = \$(1-\epsilon)$.
	\hfill $\lrcorner$
\end{exof}

Motivated by this observation we introduce a rule combining the key ideas of BOS and Add1U completion for Equal Shares.
\begin{description}
	\item[BOS Equal Shares Plus (BOS+)]
	In each round we first find the $(\alpha, \rho)$-affordable candidate $c$ that minimizes $\rho/\alpha$, as in standard BOS. 
	However, if buying candidate $c$ requires overspending, we do not buy it, but look for a better candidate that would not overspend more.
	Specifically, we first compute the maximal overspending assuming all voters overspend equally:
	\begin{align*}
		\Delta b = \frac{\cost(c) - \alpha \cost(c)}{|\{v_i \in V \colon \rho u_i(c) \geq b_i > 0\}|} \text{,}
	\end{align*}
	we temporarily set $b_i := b_i + \Delta b$ and pick the $(1, \rho)$-affordable project that fits within the budget and minimizes $\rho$. We charge the voters', revoke the increase of their entitlements, and proceed.
\end{description}
The pseudo-code of BOS+ is given in \Cref{app:bounded-overspending-plus}.

Observe that BOS+ would select ten A-projects in the instance from \Cref{ex:unpopular projects}: every attempt to 
buy a B-project 
would increase initial endowments of voters $v_1$--$v_{700}$ and allow them to buy yet another A-project.
Note also that this example highlights the fact that the optimal outcome may depend on the specific situation it is used in.
Indicatively, if the election method is used for selecting validators in a blockchain protocol, there is particular concern about not over-representing groups of voters or giving them too much voting power~\citep{cevallos2020validator}. 
Hence, BOS seems better suited for such applications. 
However, if the instance comes from participatory budgeting elections conducted by municipalities, the outcome produced by BOS+ appears to be much more aligned with expectations.	

For BOS+ we obtain the same proportionality guarantee as previously. Moreover, the hard instance from \Cref{prop:negEJRbos} also applies to BOS+.

\begin{restatable}[$\spadesuit$]{theorem}{bosPLUSejr}
	\label{thm:posPBPlus}
	For cost utilities, BOS+ satisfies EJR up to $\frac{n-|S|}{2|S|}\cdot c^*,$ where $c^*:=\max_{c\in C}\cost(c).$
\end{restatable}
\toappendix{
  \sv{\bosPLUSejr*}
  \begin{proof}
  The proof follows a similar strategy to that of \Cref{thm:posPB}. Note that both BOS and BOS+ initially identify the best-affordable project $c$ in the same way. The key difference is that BOS+ includes an additional step in which the budget is temporarily increased. As a result, the part of the proof concerning the identification of the best-affordable project $c$ can be directly reused. In particular, this applies to \Cref{1/s,half_overspend_pb}. Additionally, a significant portion of the proof of \Cref{overspendPB} remains applicable. 
  Specifically, as in the proof of \Cref{overspendPB}, we obtain that
\begin{align*}
\alpha \geq \frac{|S|}{|R|} \text{.}
\end{align*}
Now, we can bound the temporary increase in the voters' budgets, as
\begin{align*}
	\Delta b = \frac{\cost(c) - \alpha \cost(c)}{|\{v_i \in V \colon \rho u_i(c) \geq b_i > 0\}|} = \frac{\cost(c)(1 - \alpha)}{|R_{-}|} \text{.}
\end{align*}
By \Cref{half_overspend_pb} we further get that $|R_{-}| > |R|/2$ and thus
\begin{align*}
	\Delta b \leq \frac{2\cost(c)(1 - \alpha)}{|R|} \leq  \frac{2\cost(c)(1 - \frac{|S|}{|R|})}{|R|} = \frac{2\cost(c)(|R| - |S|)}{|R|^2}\leq  \frac{c^*}{2 |S|} \text{,}
\end{align*}
where the inequality follows simply by the fact that $|R|^2- 4|S|\!\cdot\! |R|+4|S|^2\geq 0$.

As with BOS, we came to the conclusion that the overall overspending of the voters is upper-bounded by  $\frac{n-|S|}{2|S|}c^*$. From there the proof follows the same way as the proof of \Cref{thm:posPB}.	
  \end{proof}
 }

\section{Empirical Analysis}
\label{sec:experiments}
Finally, we study our rules empirically
on real-world PB data from \pabulib~\citep{fal-fli-pet-pie-sko-sto-szu-tal:pb-experiments}
and on synthetic Euclidean elections.

\subsection{Pabulib Instances}
We computed the outcomes of the Utilitarian Method, Equal Shares (with Add1U completion),
FrES (with utilitarian completion), BOS, and BOS+ across all 1274 participatory budgeting instances in \pabulib.
To compare performance, we analyzed six statistics:
\myemph{score satisfaction}, \myemph{cost satisfaction}, \myemph{exclusion ratio}, \myemph{running time}, \myemph{EJR+ violations}, and \myemph{exhaustiveness}.
For the first four, we observed that the behavior largely depends on the instance size,
defined as the number of projects within it.
Thus, we partitioned the instances into four size ranges,
aiming for ranges with an almost equal number of instances.
In \Cref{fig:experiments:pabulib} we present the average of each metric by rule and size range.
The aggregated values for all six metrics can be found in \Cref{tab:pabulib:aggregated}.

\begin{table}
	\centering
	\begin{tabular}{l ccccc}
		\toprule
		Metric & Util. & Eq. Shares & BOS & BOS+ & FrES \\
		\midrule
		Avg. rel. cost satisfaction & 1.000 & 0.836 & 0.903 & 0.909 & 0.921 \\ 
		% \midrule
		% \rightcell{\quad instances with approval ballots} & 1.000 & 0.846 & 0.890 & 0.900 & 0.960 \\ 
		% \rightcell{\quad instances with choose-1 ballots} & 1.000 & 0.908 & 0.984 & 0.984 & 0.825 \\ 
		% \rightcell{\quad instances with cumulative ballots} & 1.000 & 0.751 & 0.890 & 0.887 & 0.885 \\ 
		% \rightcell{\quad instances with ordinal ballots} & 1.000 & 0.821 & 0.896 & 0.902 & 0.885 \\ 
		% \midrule
		Avg. rel. score satisfaction & 1.000 & 1.201 & 1.160 & 1.169 & 1.029 \\ 
		% \midrule
		% \rightcell{\quad instances with approval ballots} & 1.000 & 1.286 & 1.201 & 1.215 & 1.052 \\ 
		% \rightcell{\quad instances with choose-1 ballots} & 1.000 & 0.964 & 1.004 & 1.004 & 0.855 \\ 
		% \rightcell{\quad instances with cumulative ballots} & 1.000 & 1.099 & 1.138 & 1.140 & 1.070 \\ 
		% \rightcell{\quad instances with ordinal ballots} & 1.000 & 1.187 & 1.164 & 1.171 & 1.060 \\ 
		% \midrule
		Avg. exclusion ratio & 19.85\% & 17.62\% & 16.16\% & 16.50\% & 0.00\% \\ 
		% \midrule
		% \rightcell{\quad instances with approval ballots} & 13.92\% & 11.86\% & 10.36\% & 10.88\% & 0.00\% \\ 
		% \rightcell{\quad instances with choose-1 ballots} & 41.36\% & 43.37\% & 41.19\% & 41.19\% & 0.00\% \\ 
		% \rightcell{\quad instances with cumulative ballots} & 29.57\% & 27.12\% & 24.90\% & 24.89\% & 0.00\% \\ 
		% \rightcell{\quad instances with ordinal ballots} & 12.38\% & 5.10\% & 5.54\% & 5.85\% & 0.00\% \\ 
		% \midrule
		Avg. running time in sec. & 0.001 & 6.822 & 0.086 & 0.263 & 2.151 \\ 
		% \midrule
		% \rightcell{\quad instances with approval ballots} & 0.001 & 10.585 & 0.120 & 0.382 & 3.296 \\ 
		% \rightcell{\quad instances with choose-1 ballots} & 0.000 & 1.087 & 0.044 & 0.116 & 0.028 \\ 
		% \rightcell{\quad instances with cumulative ballots} & 0.000 & 0.751 & 0.014 & 0.031 & 0.060 \\ 
		% \rightcell{\quad instances with ordinal ballots} & 0.000 & 2.914 & 0.060 & 0.151 & 1.676 \\ 
		% \midrule
		Avg. EJR+ violations & 0.953 & 0.000 & 0.061 & 0.060 & 0.000 \\ 
		% \midrule
		% \rightcell{\quad instances with approval ballots} & 1.138 & 0.000 & 0.059 & 0.058 & 0.000 \\ 
		% \rightcell{\quad instances with choose-1 ballots} & 0.141 & 0.000 & 0.071 & 0.071 & 0.000 \\ 
		% \midrule
		EJR+ violation instances & 26.20\% & 0.00\% & 4.59\% & 4.48\% & 0.00\% \\ 
		% \midrule
		% \rightcell{\quad instances with approval ballots} & 29.62\% & 0.00\% & 4.29\% & 4.16\% & 0.00\% \\ 
		% \rightcell{\quad instances with choose-1 ballots} & 11.18\% & 0.00\% & 5.88\% & 5.88\% & 0.00\% \\ 
		% \midrule
		Avg. budget spending & 96.47\% & $44.84\%^{\dagger}$ & 93.98\% & 93.95\% & $93.40\%^{\dagger}$ \\ 
		% \midrule
		% \rightcell{\quad instances with approval ballots} & 96.09\% & 53.86\% & 93.20\% & 93.16\% & 92.51\% \\ 
		% \rightcell{\quad instances with choose-1 ballots} & 97.40\% & 14.98\% & 96.35\% & 96.35\% & 95.94\% \\ 
		% \rightcell{\quad instances with cumulative ballots} & 95.79\% & 27.08\% & 92.77\% & 92.77\% & 90.90\% \\ 
		% \rightcell{\quad instances with ordinal ballots} & 98.16\% & 56.85\% & 96.64\% & 96.59\% & 98.02\% \\ 
		% \midrule
		Exhausted budgets & 100.00\% & $6.99\%^{\dagger}$ & 99.84\% & 99.69\% & $24.65\%^{\dagger}$ \\ 
		% \midrule
		% \rightcell{\quad instances with approval ballots} & 100.00\% & 8.31\% & 99.73\% & 99.60\% & 17.29\% \\ 
		% \rightcell{\quad instances with choose-1 ballots} & 100.00\% & 4.12\% & 100.00\% & 100.00\% & 65.88\% \\ 
		% \rightcell{\quad instances with cumulative ballots} & 100.00\% & 4.00\% & 100.00\% & 99.50\% & 16.00\% \\ 
		% \rightcell{\quad instances with ordinal ballots} & 100.00\% & 7.59\% & 100.00\% & 100.00\% & 25.95\% \\ 
		\bottomrule
	\end{tabular}
	\caption{Aggregated statistics from running our rules on the instances from \pabulib.
	The values for Equal Shares assume Add1U completion, except for the average budget spending and exhausted budgets.
	Similarly, FrES is completed in a utilitarian fashion except for these two cases. The usage of a completion method is denoted by $\dagger$.}
	\label{tab:pabulib:aggregated}
\end{table}

Detailed data on score and cost satisfaction as well as
exclusion ratio can be found in \Cref{app:experiments}.
There, we provide mean, standard deviation, and five quantiles of values for each datapoint
shown in the first three plots of \Cref{fig:experiments:pabulib},
together with the statistical significance of the differences in values for each pair of rules.
We also include the aggregated results for all statistics
broken down by the types of ballots used in the voting process.

\begin{figure*}[t]
\pgfdeclareplotmark{mystar}{
    \node[star,star point ratio=2.25,minimum size=6pt,
          inner sep=0pt,draw=none,solid,fill=Plum] {};
}
\pgfdeclareplotmark{mydiamond}{
    \node [diamond, minimum size=0.2cm, fill=NavyBlue, inner sep=0pt, aspect=2] {};
}
\begin{center}
\hspace{-1cm}
\begin{tikzpicture}
    \begin{axis}[
        width=6.5cm,
        height=4.5cm,
        xlabel={number of projects},
        ylabel={rel. score sat.},
        xlabel style={yshift=5pt},   % Adjust the shift of xlabel
        ylabel style={yshift=-5pt},   % Adjust the shift of ylabel
        xtick={1,2,3,4},
        xticklabels={{1--8}, {9--15}, {16--27}, {28+}},
        ytick={1.0, 1.1, 1.2, 1.3, 1.4}, % Ensure specific ticks
        yticklabel style={/pgf/number format/fixed,/pgf/number format/precision=2,/pgf/number format/zerofill},
        legend pos=north east,
    %    ymajorgrids=true,
    %    xmajorgrids=true,
    %    grid style=dashed,
        legend cell align={left},
        legend style={nodes={scale=0.8, transform shape}},
        tick label style={font=\scriptsize},
        label style={font=\small},
        legend style={font=\small}
    ]
    
    \addplot[color=BrickRed,mark=*,thick,error bars/.cd,y dir=both,y explicit]
    table[x expr=\thisrowno{0}-0.10, y index=1, y error index=2] {data/score_sat.dat};
    
    \addplot[color=NavyBlue,mark=mydiamond,thick,error bars/.cd,y dir=both,y explicit]
    table[x expr=\thisrowno{0}-0.05, y index=3, y error index=4] {data/score_sat.dat};
    
    \addplot[color=OliveGreen,mark=square*,thick,error bars/.cd,y dir=both,y explicit]
    table[x expr=\thisrowno{0}, y index=5, y error index=6] {data/score_sat.dat};
    
    \addplot[color=YellowGreen,mark=triangle*,thick,error bars/.cd,y dir=both,y explicit]
    table[x expr=\thisrowno{0}+0.05, y index=7, y error index=8] {data/score_sat.dat};
    
    \addplot[color=Plum,mark=mystar,thick,error bars/.cd,y dir=both,y explicit]
    table[x expr=\thisrowno{0}+0.10, y index=9, y error index=10] {data/score_sat.dat};
    
    \end{axis}
\end{tikzpicture}
\qquad
\begin{tikzpicture}
    \begin{axis}[
        width=6.5cm,
        height=4.5cm,
        xlabel={number of projects},
        ylabel={rel. cost sat.},
        xlabel style={yshift=5pt},   % Adjust the shift of xlabel
        ylabel style={yshift=-5pt},   % Adjust the shift of ylabel
        xtick={1,2,3,4},
        xticklabels={{1--8}, {9--15}, {16--27}, {28+}},
        ytick={0.80,0.85,0.90,0.95,1.00}, % Ensure specific ticks
        yticklabel style={/pgf/number format/fixed,/pgf/number format/precision=2,/pgf/number format/zerofill},
        legend pos=north east,
    %    ymajorgrids=true,
    %    xmajorgrids=true,
    %    grid style=dashed,
        legend cell align={left},
        legend style={nodes={scale=0.8, transform shape}},
        tick label style={font=\scriptsize},
        label style={font=\small},
        legend style={font=\small}
    ]
    
    \addplot[color=BrickRed,mark=*,thick,error bars/.cd,y dir=both,y explicit]
    table[x expr=\thisrowno{0}-0.10, y index=1, y error index=2] {data/cost_sat.dat};
    %\addlegendentry{Utilitarian}
    
    \addplot[color=NavyBlue,mark=mydiamond,thick,error bars/.cd,y dir=both,y explicit]
    table[x expr=\thisrowno{0}-0.05, y index=3, y error index=4] {data/cost_sat.dat};
    %\addlegendentry{Equal Shares}
    
    \addplot[color=OliveGreen,mark=square*,thick,error bars/.cd,y dir=both,y explicit]
    table[x expr=\thisrowno{0}, y index=5, y error index=6] {data/cost_sat.dat};
    %\addlegendentry{BOS}
    
    \addplot[color=YellowGreen,mark=triangle*,thick,error bars/.cd,y dir=both,y explicit]
    table[x expr=\thisrowno{0}+0.05, y index=7, y error index=8] {data/cost_sat.dat};
    %\addlegendentry{BOS+}
    
    \addplot[color=Plum,mark=mystar,thick,error bars/.cd,y dir=both,y explicit]
    table[x expr=\thisrowno{0}+0.10, y index=9, y error index=10] {data/cost_sat.dat};
    %\addlegendentry{FrES}
    
    \end{axis}
\end{tikzpicture}
\end{center}
\begin{center}
\hspace{-1cm}
\begin{tikzpicture}
    \begin{axis}[
        width=6.5cm,
        height=4.5cm,
        xlabel={number of projects},
        ylabel={exclusion ratio},
        xlabel style={yshift=5pt},   % Adjust the shift of xlabel
        ylabel style={yshift=-5pt},   % Adjust the shift of ylabel
        xtick={1,2,3,4},
        xticklabels={{1--8}, {9--15}, {16--27}, {28+}},
        ymin=0.075, ymax=0.325,
        ytick={0.10, 0.15, 0.20, 0.25, 0.30}, % Ensure specific ticks
        yticklabel style={/pgf/number format/fixed,/pgf/number format/precision=2,/pgf/number format/zerofill},
        legend pos=north east,
    %    ymajorgrids=true,
    %    xmajorgrids=true,
    %    grid style=dashed,
        legend cell align={left},
        legend style={nodes={scale=0.8, transform shape}},
        tick label style={font=\scriptsize},
        label style={font=\small},
        legend style={font=\small}
    ]
    
    \addplot[color=BrickRed,mark=*,thick,error bars/.cd,y dir=both,y explicit]
    table[x expr=\thisrowno{0}-0.10, y index=1, y error index=2] {data/exclusion_ratio.dat};
    
    \addplot[color=NavyBlue,mark=mydiamond,thick,error bars/.cd,y dir=both,y explicit]
    table[x expr=\thisrowno{0}-0.05, y index=3, y error index=4] {data/exclusion_ratio.dat};
    
    \addplot[color=OliveGreen,mark=square*,thick,error bars/.cd,y dir=both,y explicit]
    table[x expr=\thisrowno{0}, y index=5, y error index=6] {data/exclusion_ratio.dat};
    
    \addplot[color=YellowGreen,mark=triangle*,thick,error bars/.cd,y dir=both,y explicit]
    table[x expr=\thisrowno{0}+0.05, y index=7, y error index=8] {data/exclusion_ratio.dat};
    
    \end{axis}
\end{tikzpicture}
\qquad
\begin{tikzpicture}
    \begin{axis}[
        width=6.5cm,
        height=4.5cm,
        xlabel={number of projects},
        ylabel={running time (s.)},
        xlabel style={yshift=5pt},   % Adjust the shift of xlabel
        ylabel style={yshift=-5pt},   % Adjust the shift of ylabel
        xtick={1,2,3,4},
        xticklabels={{1--8}, {9--15}, {16--27}, {28+}},
        ymode=log,
        ytick={0.0001, 0.001, 0.01, 0.1, 1, 10, 100}, % Ensure specific ticks
        legend pos=north east,
        %    ymajorgrids=true,
        %    xmajorgrids=true,
        %    grid style=dashed,
        legend cell align={left},
        legend style={nodes={scale=0.8, transform shape}},
        tick label style={font=\scriptsize},
        label style={font=\small},
        legend style={font=\small}
        ]
    
    \addplot[color=BrickRed,mark=*,thick,error bars/.cd,y dir=both,y explicit]
    table[x expr=\thisrowno{0}-0.10, y index=1, y error index=2] {data/time.dat};
    
    \addplot[color=NavyBlue,mark=mydiamond,thick,error bars/.cd,y dir=both,y explicit]
    table[x expr=\thisrowno{0}-0.05, y index=3, y error index=4] {data/time.dat};
    
    \addplot[color=OliveGreen,mark=square*,thick,error bars/.cd,y dir=both,y explicit]
    table[x expr=\thisrowno{0}, y index=5, y error index=6] {data/time.dat};
    
    \addplot[color=YellowGreen,mark=triangle*,thick,error bars/.cd,y dir=both,y explicit]
    table[x expr=\thisrowno{0}+0.05, y index=7, y error index=8] {data/time.dat};
				
    \addplot[color=Plum,mark=mystar,thick,error bars/.cd,y dir=both,y explicit]
    table[x expr=\thisrowno{0}+0.10, y index=9, y error index=10] {data/time.dat};
    
    \end{axis}
\end{tikzpicture}
\qquad
\end{center}
\begin{center}
\hspace{-2cm}
\begin{tikzpicture}
    \draw[white] (-2.1,0) rectangle (0,0);
    \draw (0,0) rectangle (12,0.5);
    \def\y{0}

    \node (util_s) at (0.15, \y + 0.25) {};
    \node (util_e) at (1.8, \y + 0.25) {\small Utilitarian};
    \node [circle, minimum size=0.175cm, fill=BrickRed, inner sep=0pt] (_) at (0.58, \y + 0.25) {};
    \node (mes_s) at (2.9, \y + 0.25) {};
    \node (mes_e) at (4.6, \y + 0.25) {\small Eq. Shares};
    \node [diamond, minimum size=0.2cm, fill=NavyBlue, inner sep=0pt, aspect=2] (_) at (3.35, \y + 0.25) {};
    \node (bos_s) at (5.8, \y + 0.25) {};
    \node (bos_e) at (7.1, \y + 0.25) {\small BOS};
    \node [regular polygon,regular polygon sides=4, minimum size=0.225cm, fill=OliveGreen, inner sep=0pt] (_) at (6.305, \y + 0.25) {};
    \node (bosp_s) at (7.8, \y + 0.25) {};
    \node (bosp_e) at (9.2, \y + 0.25) {\small BOS+};
    \node [isosceles triangle,
        minimum width=0.18cm,
        minimum height=0.17cm,
        fill=YellowGreen,
        inner sep=0pt,
        shape border rotate=90,
        isosceles triangle stretches]
        (_) at (8.29, \y + 0.238) {};
    \node (fres_s) at (9.9, \y + 0.25) {};
    \node (fres_e) at (11.2, \y + 0.25) {\small FrES};
    \node[star,star point ratio=2.25,minimum size=6pt,
        inner sep=0pt,draw=none,solid,fill=Plum]
        (_) at (10.39, \y + 0.258) {};

    \path[thick]
    (util_s) edge[draw=BrickRed] (util_e)
    (mes_s) edge[draw=NavyBlue] (mes_e)
    (bos_s) edge[draw=OliveGreen] (bos_e)
    (bosp_s) edge[draw=YellowGreen] (bosp_e)
    (fres_s) edge[draw=Plum] (fres_e)
    ;
\end{tikzpicture}
\end{center}
	\caption{
        Average 
		score satisfaction (top-left),
		cost satisfaction (top-right),
        exclusion ratio (bottom-right), and
        running time in seconds (bottom-left)
		of the outputs of the considered rules
		based on \pabulib{} data by the number projects in an instance.
		Cost and score satisfactions are reported as fractions
		of the respective satisfaction under utilitarian rule.
		Thin vertical lines mark 90\% confidence intervals.
        Note that the $y$-axis in the running time plot is in a logarithmic scale.}
	\label{fig:experiments:pabulib}
\end{figure*}

\paragraph{Score and Cost Satisfaction.}
These metrics assess the total utility from the selected projects based on score and cost utility measures.
In \Cref{fig:experiments:pabulib}, the values are normalized against Utilitarian.
For score satisfaction, in medium and large instances, Equal Shares, BOS, and BOS+ yield similar and significantly higher results than FrES and Utilitarian; for small instances, the performance of all is comparable.
Regarding cost satisfaction, Utilitarian outperforms all others  across all instances. Indeed, it can be considered a greedy method maximizing cost satisfaction objective.
Yet, importantly, BOS and FrES outperform Equal Shares in small and medium instances.
Notably, BOS+ performs as good as BOS for the smaller instances and slightly surpasses it for the larger.

\paragraph{Exclusion Ratio.}
This metric represents the fraction of voters who do not support any of the selected projects.
For medium and large instances
Utilitarian excludes significantly more voters on average
than the other rules, which perform similarly. 
However, for the smallest instances, Equal Shares performs considerably worse than the others. 
This is due to small instances often resembling the Helenka Paradox.
For FrES, the exclusion ratio is always zero, 
so it is not included. 

\paragraph{Running Time.}
We observe that BOS and BOS+ are an order of magnitude faster than Equal Shares.
This is because Equal Shares with Add1U completion,
requires computing the outcome multiple times
until the final value of voter endowments is established.
The running time for FrES
grows much more sharply with the instance size increase
than for the other rules.
This can be explained by the fact that
FrES involves making multiple decisions about fractional purchases
rather than selecting whole projects.
The running times reported here are based on computation in Python
on a MacBook Air laptop with an Apple M3 processor and 16 GB of RAM.

\paragraph{EJR+ Violations.}
EJR+ is one of the strongest satisfiable proportionality axioms, strictly stronger than EJR~\citep{BP-ejrplus}.
Its satisfaction can be efficiently verified,
making it a suitable metric for empirical comparison of the proportionality of different rules.
Unfortunately, to date, it is defined only for approval ballots,
thus we limit our analysis to such instances.
On a high level, for each unselected candidate in an election,
we say that it violates EJR+ up to one
(we actually count violations of this established relaxation~\citep{BP-ejrplus}),
if it is approved by a sufficiently large group of voters
each of whom is inadequately satisfied by the outcome of a rule.
For Equal Shares, it is known that its outcome cannot have any violations.
In contrast, Utilitarian averages 0.933 violations per instance,
setting it apart from BOS and BOS+, averaging only 0.061 and 0.06, respectively.
This indicates that although our rules, in theory, do not guarantee EJR+ up to one,
they almost always yield proportional solutions according to this strong axiom.
In \Cref{tab:pabulib:aggregated}, for each rule we also provide
in how many instances there is at least one EJR+ violation.

\paragraph{Exhaustiveness}
Recall that the outcome of a rule is exhaustive,
if there is no unselected project with a cost not exceeding the remaining budget
(for FrES, we require that it either spends the whole budget
or selects all projects fully).
By definition, the Utilitarian Method is always exhaustive.
In contrast, if we run Equal Shares without any completion,
then it is exhaustive in only 7\% of the instances
and on average spends less than 45\% of the available budget.
On the other hand, BOS is non-exhaustive in
only 2 instances out of 1274 in the dataset.
Moreover, in both of these instances,
it spends more than 99\% of the budget
(only very cheap additional projects can fit the budget in both cases).
BOS+ is non-exhaustive in a few more instances,
but spends similar fraction of the budget on average.
As a result, while Equal Shares should not be used without a completion mechanism,
BOS and BOS+ can safely be deployed without one.
Finally, we note that FrES is rarely exhaustive.
This is typically caused by voters that support single projects.
If FrES selects such projects fully,
the remaining funds of their supporters will not be spent.

\subsection{Euclidean Instances}

We further illustrate the differences between the examined rules using synthetic Euclidean elections.
Our results are presented in \Cref{fig:experiments:euclidean}. 
This analysis again highlights that BOS offers more desirable results than Equal Shares.

\begin{figure}[t]
	\centering
	\setlength{\tabcolsep}{-1.2pt}
	\renewcommand{\arraystretch}{-6}
	\begin{tabular}{cccccc}
		{\small Voters} & {\small Utilitarian} & {\small Equal Shares} & {\small BOS} & {\small BOS+} & {\small FrES} \\
		\includegraphics[width=2.25cm]{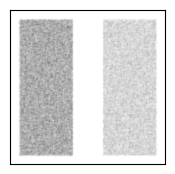} &
		\includegraphics[width=2.25cm]{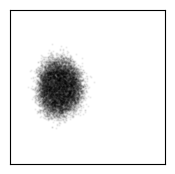} &
		\includegraphics[width=2.25cm]{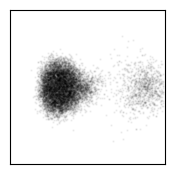} &
		\includegraphics[width=2.25cm]{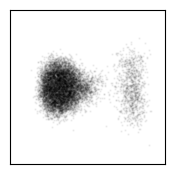} &
		\includegraphics[width=2.25cm]{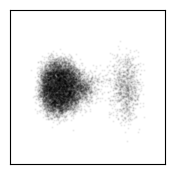} &
		\includegraphics[width=2.25cm]{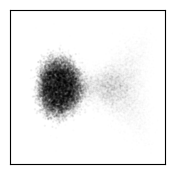} \\
		\includegraphics[width=2.25cm]{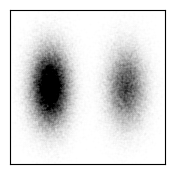} &
		\includegraphics[width=2.25cm]{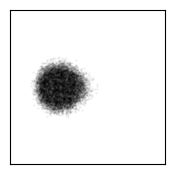} &
		\includegraphics[width=2.25cm]{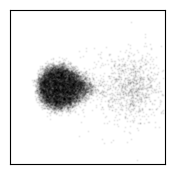} &
		\includegraphics[width=2.25cm]{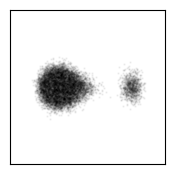} &
		\includegraphics[width=2.25cm]{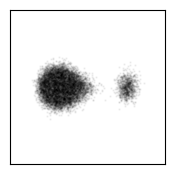} &
		\includegraphics[width=2.25cm]{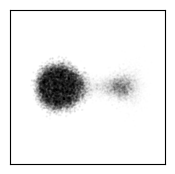} \\
		\includegraphics[width=2.25cm]{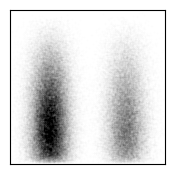} &
		\includegraphics[width=2.25cm]{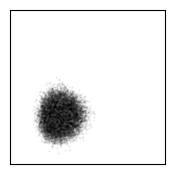} &
		\includegraphics[width=2.25cm]{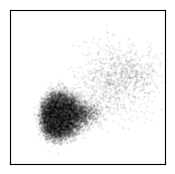} &
		\includegraphics[width=2.25cm]{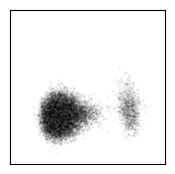} &
		\includegraphics[width=2.25cm]{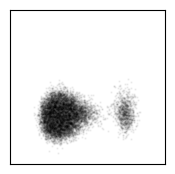} &
		\includegraphics[width=2.25cm]{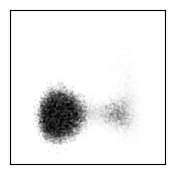} \\
	\end{tabular}
	\caption{Results on Euclidean elections.
		The first column presents the superimposed positions of voters in all generated elections and the following columns show positions of candidates selected by respective rules
		(for FrES, the opacity reflects the selected fraction of candidates).}
	\label{fig:experiments:euclidean}
\end{figure}

In the examined setting, each voter and candidate is represented as a point in a 2D space.
The utility of the voter $ v_i $ for candidate $ c $
is defined as $ u_i(c) = (\dist(v_i, c) + \lambda)^{-1} $,
where $ \dist(v_i, c) $ is the Euclidean distance between them, and
the denominator is shifted by the constant of $\lambda$
to bound the maximal utility.
Here, we present the results for $\lambda = 1$;
the results for $\lambda = \nicefrac{1}{2}$ and $\lambda = 2$
(in \Cref{app:experiments}) lead to similar conclusion.
For each election, we placed 150 candidates uniformly at random in a unit square
and considered three different voter distributions.
In the first, 100 voters were drawn uniformly from the rectangle $(0.05, 0.05)-(0.4, 0.95)$
and 50 from $(0.6, 0.05)-(0.95, 0.95)$. 
The second and third used Gaussian distributions for $x$-coordinates:
100 voters centered at 0.25 and 50 centered at 0.75. 
In the second distribution, $y$-coordinates were drawn from a Gaussian centered at 0.5,
while in the third, they were drawn from a beta distribution with $a=1.5$ and $b=3$,
pushing more voters to the bottom of the square. 
The candidates had unit costs, and the budget was set to 10.
We sampled 1000 elections per distribution.

The Utilitarian Method selects only candidates from the left-hand side of the square.
Equal Shares, BOS, BOS+, and FrES also select candidates from the right-hand side
but less than 1/3 of the total, which is the proportion of the voters there. 
This is because the right-hand side voters gain some utility from left-hand side candidates,
thus spending part of their budget on them. 
In the third distribution,
Equal Shares selects candidates from the top-right part of the square,
which has fewer voters, while BOS and BOS+ choose candidates from the bottom-right,
where there is greater overall support.
BOS+ and BOS generally output very similar sets of candidates,
with BOS+ giving a bit more concentrated outputs.
FrES's selection resembles that of BOS,
particularly in the last two distributions,
reinforcing the view that BOS is a rounding of FrES.

\section{Conclusion}

We have introduced the Method of Equal Shares with Bounded Overspending (BOS), a robust variant of the Method of Equal Shares, along with its refined version, BOS+. We have identified inefficiencies in the original method, illustrated by simple well-structured examples and supported by the analysis of real-life data. Specifically, we have shown that Equal Shares underperforms in small PB instances. BOS effectively handles all types of instances, in particular solving the problematic issues of Equal Shares.
Our new rule maintains strong fairness properties, which we have confirmed both theoretically and in experiments. 
In the process of developing BOS, we have also introduced and analyzed FrES, a fractional variant of Equal Shares for the PB model with additive utilities. 
Future research will focus on further theoretical guarantees for BOS and BOS+.

\section*{Acknowledgements}
T. \Was{} was supported by UK Engineering and Physical Sciences Research Council (EPSRC) under grant EP/X038548/1. The rest of the authors were partially supported by the European Union (ERC, PRO-DEMOCRATIC, 101076570). Views and opinions expressed are however those of the authors only and do not necessarily reflect those of the European Union or the European Research Council. Neither the European Union nor the granting authority can be held responsible for them.

% \newpage	
\bibliographystyle{apalike}
\bibliography{simplified-mes}

\clearpage
\appendix
\appendixsectionformat % Activate custom formatting for appendix sections

\section{Complete Examples Illustrating the Rules’ Execution}
\label{app:examples}
In this section, we provide detailed description
of all steps in the runs of Equal Shares, FrES, and BOS
on the instance from \Cref{ex:mes}.
Let us repeat the table with the projects and voters for convenience.

\begin{center}
	\begin{tabularx}{12cm}{lc|XXXXXXXXXX}
		& cost & $v_1$  & $v_2$ & $v_3$ & $v_4$ & $v_5$ & $v_6$ & $v_7$ & $v_8$ & $v_9$ & $v_{10}$ \\
		\midrule
		Project A 
		& \$300k & \faCheck & \faCheck & \faCheck & \faCheck & \faCheck & \faCheck & & & & \\
		Project B  
		& \$400k & & \faCheck & \faCheck & \faCheck & \faCheck & \faCheck & & & & \\
		Project C  
		& \$300k & & \faCheck & \faCheck & \faCheck  & \faCheck  & & & & & \faCheck \\
		Project D 
		& \$240k & & & & & & & \faCheck & \faCheck & \faCheck & \faCheck \\
		Project E  
		& \$170k & & \faCheck & & & & & \faCheck & \faCheck & \faCheck & \\
		Project F  
		& \$100k & & & & & & \faCheck & & & \faCheck & \faCheck  \\
	\end{tabularx}
\end{center}

The Utilitarian Method would select projects solely based on their vote count, thus choosing Projects A, B and C. 
This seems unfair since a large fraction of the voters (namely voters $v_7$ to $v_9$ making up $30\%$ of the electorate) would not approve any of the selected projects.

\subsection{Method of Equal Shares}

At the beginning, the Method of Equal Shares assigns \$100k to each voter. The table that appears below indicates the available (virtual) budget of each voter (in thousands of dollars).
\begin{center}
	\setlength{\tabcolsep}{2pt}
	\begin{tabularx}{10cm}{c|YYYYYYYYYY}
		& $v_1$ & $v_2$ & $v_3$ & $v_4$ & $v_5$ & $v_6$ & $v_7$ & $v_8$ & $v_9$ & $v_{10}$ \\
		\midrule
		\ $b_i$\ \ & \small $100$ & \small $100$ & \small $100$ & \small $100$ & \small $100$ & \small $100$ & \small $100$ & \small $100$ & \small $100$ & \small $100$ \\
	\end{tabularx}
\end{center}
We first need to determine how affordable each project is. Project A is $\nicefrac{1}{6}$-affordable, as it can be funded if each of its supporters pays \$50k, which is $\nicefrac{1}{6}$ of its cost (recall that we assumed cost-utilities for this example).
Analogously, each other project that received $x$ votes is $\nicefrac{1}{x}$-affordable.
\begin{center}
	\setlength{\tabcolsep}{2pt}
	\begin{tabularx}{5cm}{c|aYYYYY}
		& A & B & C & D & E & F \\
		\midrule
		\ $\rho$\ \  & \nicefrac{1}{6} & \nicefrac{1}{5} & \nicefrac{1}{5} & \nicefrac{1}{4} & \nicefrac{1}{4} & \nicefrac{1}{3} \\
	\end{tabularx}
\end{center}
Thus, in the first round the rule simply selects the project with the highest vote count, namely project A. After paying its cost, voters $v_1$ to $v_6$ are left with $\$(100\text{k} - 300\text{k}/6) = \$50\text{k}$. Voters' remaining budget follows.
\begin{center}
	\setlength{\tabcolsep}{2pt}
	\begin{tabularx}{10cm}{c|YYYYYYYYYY}
		& $v_1$ & $v_2$ & $v_3$ & $v_4$ & $v_5$ & $v_6$ & $v_7$ & $v_8$ & $v_9$ & $v_{10}$ \\
		\midrule
		\ $b_i$\ \ & \small $50$ & \small $50$ & \small $50$ & \small $50$ & \small $50$ & \small $50$ & \small $100$ & \small $100$ & \small $100$ & \small $100$ \\
	\end{tabularx}
\end{center}

In the second round, project B is no longer affordable, as its supporters do not have a total of at least \$400k to fund it. Project C is $\nicefrac{1}{3}$-affordable. Indeed, to fund it, voters $v_2$, $v_3$, $v_4$, $v_5$ and $v_{10}$ would have to use all their money. This means, in particular, that voter $v_{10}$ would pay \$100k out of \$300k which is $\nicefrac{1}{3}$ of the cost of the project. The rest of the projects remain $\rho$-affordable for the same values of $\rho$, as their cost can still be spread equally among their supporters. 
\begin{center}
	\setlength{\tabcolsep}{2pt}
	\begin{tabularx}{5cm}{c|aYYaYY}
		& A & B & C & D & E & F \\
		\midrule
		\ $\rho$\ \ & -- & -- & \nicefrac{1}{3} & \nicefrac{1}{4} & \nicefrac{1}{4} & \nicefrac{1}{3} \\
	\end{tabularx}
\end{center}
Hence, projects D and E are both $\rho$-affordable for the smallest value of $\rho = \nicefrac{1}{4}$. Let us assume the former project is selected, by breaking ties lexicographically. After paying its cost, voters $v_7$ to $v_{10}$ are left with $\$(100\text{k} - 240\text{k}/4) = \$40\text{k}$.
\begin{center}
	\setlength{\tabcolsep}{2pt}
	\begin{tabularx}{10cm}{c|YYYYYYYYYY}
		& $v_1$ & $v_2$ & $v_3$ & $v_4$ & $v_5$ & $v_6$ & $v_7$ & $v_8$ & $v_9$ & $v_{10}$ \\
		\midrule
		\ $b_i$\ \  & \small $50$ & \small $50$ & \small $50$ & \small $50$ & \small $50$ & \small $50$ & \small $40$ & \small $40$ & \small $40$ & \small $40$ \\
	\end{tabularx}
\end{center}

In the third round, project C is no longer affordable, as its supporters have together only \$240k. Project E is $\nicefrac{5}{17}$-affordable: voters $v_7$, $v_8$ and $v_9$ have only \$120k, which means that voter $v_2$ needs to pay \$50k to cover the cost of \$170k.
Project F remains $\nicefrac{1}{3}$-affordable.
\begin{center}
	\setlength{\tabcolsep}{2pt}
	\begin{tabularx}{5cm}{c|aYYaaY}
		& A & B & C & D & E & F \\
		\midrule
		\ $\rho$\ \ & -- & -- & -- & -- & \nicefrac{5}{17} & \nicefrac{1}{3} \\
	\end{tabularx}
\end{center}
As a result, project E is selected and voters $v_2$, $v_7$, $v_8$, and $v_9$ run out of money.
\begin{center}
	\setlength{\tabcolsep}{2pt}
	\begin{tabularx}{10cm}{c|YYYYYYYYYY}
		& $v_1$ & $v_2$ & $v_3$ & $v_4$ & $v_5$ & $v_6$ & $v_7$ & $v_8$ & $v_9$ & $v_{10}$ \\
		\midrule
		\ $b_i$\ \  & \small $50$ & \small $0$ & \small $50$ & \small $50$ & \small $50$ & \small $50$ & \small $0$ & \small $0$ & \small $0$ & \small $40$ \\
	\end{tabularx}
\end{center}

At this point, no project is affordable since the supporters of projects B, C and F have in total \$200k, \$150k and \$90k, respectively, therefore, the procedure stops having selected the outcome $\{A,D,E\}$. Note that the purchased bundle comes at a total cost of \$710k, which is \$290k less than the initially available budget. Thus, in principle, we could afford to additionally fund project F. However, the supporters of this project do not have enough (virtual) money to fund it, and so the project is not selected by the Method of Equal Shares. 

Clearly, the selection made by the Method of Equal Shares is less discriminatory than the one by Utilitarian, as each voter approves at least one of the selected projects.

\subsection{Fractional Equal Shares}

Note that for cost utilities, in each round Fractional Equal Shares selects the project with the most supporters who still have money left.
The method purchases the largest portion of the project that can be covered with equal payments of the supporters.

In the first round, as in the Method of Equal Shares, project A is chosen. It is bought in full, as its whole price can be split equally between supporters, each paying \$50k.
Voters' remaining funds appear below.
\begin{center}
	\setlength{\tabcolsep}{2pt}
	\begin{tabularx}{10cm}{c|YYYYYYYYYY}
		& $v_1$ & $v_2$ & $v_3$ & $v_4$ & $v_5$ & $v_6$ & $v_7$ & $v_8$ & $v_9$ & $v_{10}$ \\
		\midrule
		\ $b_i$\ \ & \small $50$ & \small $50$ & \small $50$ & \small $50$ & \small $50$ & \small $50$ & \small $100$ & \small $100$ & \small $100$ & \small $100$ \\
	\end{tabularx}
\end{center}

In the second round, projects B or C can be selected. Let us assume project C is selected. Since voters $v_2$ to $v_5$ have only \$50k left, only $\nicefrac{5}{6}$ of it is bought and each voter pays \$50k as the following table indicates.
\begin{center}
	\setlength{\tabcolsep}{2pt}
	\begin{tabularx}{6cm}{c|aYaYYY}
		& A & B & C & D & E & F \\
		\midrule
		\ $\rho$\ \  & -- & \nicefrac{1}{5} & \nicefrac{1}{5} & \nicefrac{1}{4} & \nicefrac{1}{4} & \nicefrac{1}{3} \\
		\midrule
		\ $\alpha$\ \  & 1 & & $\nicefrac{5}{6}$ \\
	\end{tabularx}
\end{center}
The following table depicts the remaining amount of money of each voter.
\begin{center}
	\setlength{\tabcolsep}{2pt}
	\begin{tabularx}{10cm}{c|YYYYYYYYYY}
		& $v_1$ & $v_2$ & $v_3$ & $v_4$ & $v_5$ & $v_6$ & $v_7$ & $v_8$ & $v_9$ & $v_{10}$ \\
		\midrule
		\ $b_i$\ \ & \small $50$ & \small $0$ & \small $0$ & \small $0$ & \small $0$ & \small $50$ & \small $100$ & \small $100$ & \small $100$ & \small $50$ \\
	\end{tabularx}
\end{center}

Since voters $v_2$ to $v_5$ run out of money, in the third round project B has only one vote and project E has 3 votes instead of 4. Hence, project D with 4 votes is selected.
\begin{center}
	\setlength{\tabcolsep}{2pt}
	\begin{tabularx}{6cm}{c|aYaaYY}
		& A & B & C & D & E & F \\
		\midrule
		\ $\rho$\ \  & -- & \nicefrac{1}{1} & \nicefrac{1}{1} & \nicefrac{1}{4} & \nicefrac{1}{3} & \nicefrac{1}{2} \\
		\midrule
		\ $\alpha$\ \  & 1 & & $\nicefrac{5}{6}$ & \nicefrac{5}{6} \\
	\end{tabularx}
\end{center}
For buying project D all voters pay the maximal equal price: \$50k and their remaining budget appears below:
\begin{center}
	\setlength{\tabcolsep}{2pt}
	\begin{tabularx}{10cm}{c|YYYYYYYYYY}
		& $v_1$ & $v_2$ & $v_3$ & $v_4$ & $v_5$ & $v_6$ & $v_7$ & $v_8$ & $v_9$ & $v_{10}$ \\
		\midrule
		\ $b_i$\ \ & \small $50$ & \small $0$ & \small $0$ & \small $0$ & \small $0$ & \small $50$ & \small $50$ & \small $50$ & \small $50$ & \small $0$ \\
	\end{tabularx}
\end{center}

In the fourth round, project D lost one vote as $v_{10}$ run out of money, but still has the highest number of votes (ex aequo with project E). Hence, Fractional Equal Shares buys the remaining $\nicefrac{1}{6}$ of project D asking voters $v_7$ to $v_9$ to pay \$13.$\bar{3}$ each.
\begin{center}
	\setlength{\tabcolsep}{2pt}
	\begin{tabularx}{6cm}{c|aYaaYY}
		& A & B & C & D & E & F \\
		\midrule
		\ $\rho$\ \  & -- & \nicefrac{1}{1} & -- & \nicefrac{1}{3} & \nicefrac{1}{3} & \nicefrac{1}{2} \\
		\midrule
		\ $\alpha$\ \  & 1 & & $\nicefrac{5}{6}$ & 1 \\
	\end{tabularx}
\end{center}
After the purchase, the remaining funds are as follows.
\begin{center}
	\setlength{\tabcolsep}{2pt}
	\begin{tabularx}{10cm}{c|YYYYYYYYYY}
		& $v_1$ & $v_2$ & $v_3$ & $v_4$ & $v_5$ & $v_6$ & $v_7$ & $v_8$ & $v_9$ & $v_{10}$ \\
		\midrule
		\ $b_i$\ \ & \small $50$ & \small $0$ & \small $0$ & \small $0$ & \small $0$ & \small $50$ & \small $36.\bar{6}$ & \small $36.\bar{6}$ & \small $36.\bar{6}$ & \small $0$ \\
	\end{tabularx}
\end{center}

Since the fourth round did not change the numbers of supporters as no new voters ran out of budget, a fraction of $\nicefrac{11}{17}$ of project E is bought next. 
\begin{center}
	\setlength{\tabcolsep}{2pt}
	\begin{tabularx}{6cm}{c|aYaaaY}
		& A & B & C & D & E & F \\
		\midrule
		\ $\rho$\ \  & -- & \nicefrac{1}{1} & -- & -- & \nicefrac{1}{3} & \nicefrac{1}{2} \\
		\midrule
		\ $\alpha$\ \  & 1 & & $\nicefrac{5}{6}$ & 1 & \nicefrac{11}{17} \\
	\end{tabularx}
\end{center}
Regarding voters' currently available budget, only $v_1$ and $v_6$ have a positive amount, each still having \$50k. 
\begin{center}
	\setlength{\tabcolsep}{2pt}
	\begin{tabularx}{10cm}{c|YYYYYYYYYY}
		& $v_1$ & $v_2$ & $v_3$ & $v_4$ & $v_5$ & $v_6$ & $v_7$ & $v_8$ & $v_9$ & $v_{10}$ \\
		\midrule
		\ $b_i$\ \ & \small $50$ & \small $0$ & \small $0$ & \small $0$ & \small $0$ & \small $50$ & \small $0$ & \small $0$ & \small $0$ & \small $0$ \\
	\end{tabularx}
\end{center}
As a result, projects B and F have one supporter each whose money can be used to buy a portion of one of these projects. Assume $\nicefrac{1}{2}$ of project F is bought and $v_6$ pays for it. 
\begin{center}
	\setlength{\tabcolsep}{2pt}
	\begin{tabularx}{6cm}{c|aYaaaa}
		& A & B & C & D & E & F \\
		\midrule
		\ $\rho$\ \  & -- & \nicefrac{1}{1} & -- & -- & -- & \nicefrac{1}{1} \\
		\midrule
		\ $\alpha$\ \  & 1 & & $\nicefrac{5}{6}$ & 1 & \nicefrac{11}{17} & \nicefrac{1}{2} \\
	\end{tabularx}
\end{center}

Finally, the only voter with positive amount of money is $v_1$ who is left with \$50k.
However, the only project that $v_1$ supports has already been bought.
Thus, Fraction Equal Shares concludes.

As a result, in our example, FrES allocated \$550k to projects A--C
and \$400k to projects D--F, while the Method of Equal Shares allocated \$300k to the former and \$410k to the latter.
In what follows, we will propose a modification of Equal Shares
that is more aligned with the outcomes of FrES
and in this example spends more funds on projects A--C than on D--F.
% \hfill $\lrcorner$

\subsection{BOS Equal Shares}

Consider the first round. At the beginning, each voter has the same amount of money:
\begin{center}
	\setlength{\tabcolsep}{2pt}
	\begin{tabularx}{10cm}{l|YYYYYYYYYY}
		& $v_1$ & $v_2$ & $v_3$ & $v_4$ & $v_5$ & $v_6$ & $v_7$ & $v_8$ & $v_9$ & $v_{10}$ \\
		\midrule
		\ $b_i$\ \  & \small $100$ & \small $100$ & \small $100$ & \small $100$ & \small $100$ & \small $100$ & \small $100$ & \small $100$ & \small $100$ & \small $100$
	\end{tabularx}
\end{center}			
Note that the cost of every project can be covered by its supporters.
As observed in \Cref{ex:bos_helenka}, for every project, $\alpha=1$ is optimal, which corresponds to buying the project in full.
\begin{center}
	\setlength{\tabcolsep}{2pt}
	\begin{tabularx}{7cm}{l|aYYYYY}
		& A & B & C & D & E & F \\
		\midrule
		\ $\alpha = 1, \rho$\ \ & \nicefrac{1}{6} & \nicefrac{1}{5} & \nicefrac{1}{5} & \nicefrac{1}{4} & \nicefrac{1}{4} & \nicefrac{1}{3}
	\end{tabularx}
\end{center}
Hence, BOS, as the Method of Equal Shares, selects project A, being the one with the highest number of votes and equally distributes its cost among the supporters.
After the purchase, voters have the following amount of money:
\begin{center}
	\setlength{\tabcolsep}{2pt}
	\begin{tabularx}{10cm}{l|YYYYYYYYYY}
		& $v_1$ & $v_2$ & $v_3$ & $v_4$ & $v_5$ & $v_6$ & $v_7$ & $v_8$ & $v_9$ & $v_{10}$ \\
		\midrule
		\ $b_i$\ \  & \small $50$ & \small $50$ & \small $50$ & \small $50$ & \small $50$ & \small $50$ & \small $100$ & \small $100$ & \small $100$ & \small $100$ 
	\end{tabularx}
\end{center}			

In the second round the remaining budget is $\$700\text{k}$. 	
Reasonable values of $\alpha$ and the corresponding $\rho$ are as follows:
\begin{center}
	\setlength{\tabcolsep}{2pt}
	\begin{tabularx}{7cm}{l|aYaYYY}
		& A & B & C & D & E & F \\
		\midrule
		\ $\alpha = 1, \rho$\ \  & -- & -- & \nicefrac{1}{3} & \nicefrac{1}{4} & \nicefrac{1}{4} & \nicefrac{1}{3} \\
		\ $\alpha = \nicefrac{5}{6}, \rho$\ \  & -- & -- & \nicefrac{1}{5} & \nicefrac{1}{4} & \nicefrac{1}{4} & \nicefrac{1}{3} \\
		\ $\alpha = \nicefrac{5}{8}, \rho$\ \ & -- & \nicefrac{1}{5} & \nicefrac{1}{5} & \nicefrac{1}{4} & \nicefrac{1}{4} & \nicefrac{1}{3} \\
	\end{tabularx}
\end{center}
Since voters $v_1$ to $v_6$ have only $\$50\text{k}$ each, project B cannot be purchased in full. However, its supporters can cover $\alpha = \nicefrac{250}{400} = \nicefrac{5}{8}$ of its cost. This cost would be equally spread among five voters, so $\rho = \nicefrac{1}{5}$ and the ratio ($\rho$ scaled by $\alpha$) equals $\rho/\alpha = \nicefrac{8}{25}$.

Now, let us consider project C. This project can be bought in full, but only if voter $v_{10}$ pays $\rho = \nicefrac{1}{3}$ of the cost. The Method of Equal Shares rejects this option as imbalanced and selects a project with a smaller $\rho$ parameter. In turn, Bounded Overspending considers also buying a fraction of the project with balanced payments. Specifically, to maintain equal payments, $\alpha=\nicefrac{5}{6}$ of the project can be bought. This would result in $\rho=\nicefrac{1}{5}$ and ratio $\rho/\alpha = \nicefrac{6}{25}$. Since this ratio is smaller than the ratio of project B and the ratio of projects D, E and F which remain unchanged, Bounded Overspending selects project C with $\alpha = \nicefrac{5}{6}$. To cover $\nicefrac{5}{6}$ of its cost, its supporters would have to pay $\$50\text{k}$ each. Hence, in Bounded Overspending, each voter pays $\nicefrac{6}{5} \cdot \$50\text{k} = \$60\text{k}$ (or their whole budget).
\begin{center}
	\setlength{\tabcolsep}{2pt}
	\begin{tabularx}{10cm}{l|YYYYYYYYYY}
		& $v_1$ & $v_2$ & $v_3$ & $v_4$ & $v_5$ & $v_6$ & $v_7$ & $v_8$ & $v_9$ & $v_{10}$ \\
		\midrule
		\ $b_i$\ \  & \small $50$ & \small $0$ \tiny$(-10)$ & \small $0$ \tiny$(-10)$ & \small $0$ \tiny$(-10)$ & \small $0$ \tiny$(-10)$ & \small $50$ & \small $100$ & \small $100$ & \small $100$ & \small $40$
	\end{tabularx}
\end{center}			

In the third round, the remaining budget is $\$400\text{k}$.
Consider the following $\alpha$ and $\rho$ values:
\begin{center}
	\setlength{\tabcolsep}{2pt}
	\begin{tabularx}{7cm}{l|aYaaYY}
		& A & B & C & D & E & F \\
		\midrule
		\ $\alpha = 1, \rho$\ \  & -- & -- & -- & \nicefrac{5}{18} & \nicefrac{1}{3} & \nicefrac{1}{3} \\
		\ $\alpha = \nicefrac{2}{3}, \rho$\ \  & -- & -- & -- & \nicefrac{1}{4} & \nicefrac{1}{3} & \nicefrac{1}{3} \\
		\ $\alpha = \nicefrac{1}{8}, \rho$\ \  & -- & \nicefrac{1}{1} & -- & \nicefrac{1}{4} & \nicefrac{1}{3} & \nicefrac{1}{3} \\
	\end{tabularx}
\end{center}			
Note that only $\nicefrac{1}{8}$ of the project B can be bought (with $\rho=1$). Project E and F can still be bough in full, but since voter $v_2$ run out of money, the value of $\rho$ for project E dropped to $\nicefrac{1}{3}$. 
Now, project D can either be bought in full (with unequal payments $(66.\bar{6}, 66.\bar{6}, 66.\bar{6}, 40)$ and $\rho = 66.\bar{6}/240 = \frac{5}{18}$) or in part: for $\alpha = \nicefrac{2}{3}$ with equal payments $(40,40,40,40)$ and $\rho = \frac{1}{4}$. Since the ratio $\rho/\alpha$ for the former option is lower and it is also lower than the ratio for other projects, Bounded Overspending selects project D with unequal payments.

\begin{center}
	\setlength{\tabcolsep}{2pt}
	\begin{tabularx}{10cm}{l|YYYYYYYYYY}
		& $v_1$ & $v_2$ & $v_3$ & $v_4$ & $v_5$ & $v_6$ & $v_7$ & $v_8$ & $v_9$ & $v_{10}$ \\
		\midrule
		\ $b_i$\ \  & \small $50$ & \small $0$ \tiny$(-10)$ & \small $0$ \tiny$(-10)$ & \small $0$ \tiny$(-10)$ & \small $0$ \tiny$(-10)$ & \small $50$ & \small $33.\bar{3}$ & \small $33.\bar{3}$ & \small $33.\bar{3}$ & \small $0$
	\end{tabularx}
\end{center}			

In the fifth round, the remaining budget is $\$160\text{k}$.
Consider the following values of $\alpha$ and $\rho$:
\begin{center}
	\setlength{\tabcolsep}{2pt}
	\begin{tabularx}{7cm}{l|aYaaYa}
		& A & B & C & D & E & F \\
		\midrule
		\ $\alpha = 1, \rho$\ \  & -- & -- & -- & -- & -- & -- \\
		\ $\alpha = \nicefrac{5}{6}, \rho$\ \  & -- & -- & -- & -- & -- & \nicefrac{3}{5} \\
		\ $\alpha = \nicefrac{2}{3}, \rho$\ \  & -- & -- & -- & -- & -- & \nicefrac{1}{2} \\
	\end{tabularx}
\end{center}
Hence, only project F can be bought which is done for $\alpha=\nicefrac{5}{6}$, $\rho = \nicefrac{3}{5}$ (voter $v_6$ pays $50$ out of $83.\bar{3}$ that supporters can cover) and ratio $\rho/\alpha = \frac{18}{25}$.
As a result, the remaining budget is \$60k, no project can be added and BOS terminates.
% \hfill $\lrcorner$

\clearpage
\section{Proofs Omitted from the Main Text}\label{app:proofs}
\appendixText

\clearpage

\clearpage
\section{Pseudo-Code of BOS Equal Shares Plus}
\label{app:bounded-overspending-plus}
\SetKw{Return}{return}
\SetKw{Input}{Input:}
\begin{algorithm}[h!]
	\small
	\captionsetup{labelfont={sc,bf}, labelsep=newline}
	\caption{Pseudo-code of BOS Equal Shares Plus.}\label{alg:bounded-overspending-plus}
	\Input{A PB election $(C,V,b)$ \\}
	\DontPrintSemicolon
	\SetAlgoNoEnd
	\SetAlgoLined
	$W \gets \emptyset, \quad b_i \gets \nicefrac{b}{n}  \text{~for each~} v_i \in V$ \; 
	% $b_i \gets \nicefrac{b}{n}  \text{~for each~} v_i \in V$ \;
	$\mathrm{over}_i \gets 0  \text{~for each~} v_i \in V$ {\color{winered}\tcc*[l]{Overspending for each voter}}
	\While{\emph{exists} $c \in C \setminus W$ \emph{s.t.} $\cost(c) \le b-\cost(W)$}{
		{\color{winered}\tcc{The part of computing the values $\alpha^*$, $\rho^*$, $c^*$ is the same as in pure BOS.}}
		$(\alpha^*, \rho^*, c^*) \gets (1, +\infty, c)$ \;
		
		\For{$c \in C \setminus W$ \emph{s.t.} $\cost(c) \le b-\cost(W)$}{
			$\lambda' \gets \lambda$ satisfying $\cost(c) = \sum_{i=1}^n \min(b_i, u_i(c) \cdot \lambda)$ or $+\infty$ if there is no such $\lambda$, i.e., $\cost(c)>\sum_{v_i \in V : u_i(c) > 0}b_i$\;
			\For{$\lambda \in \{b_i/u_i(c) : v_i \in V, b_i >0, u_i(c) > 0\} \cup \{\lambda'\}$}{
				$\alpha \gets \min(\left(\sum_{i=1}^n \min(b_{i},u_{i}(c) \cdot \lambda)\right)/\cost(c),1)$ \;
				$\rho \gets \lambda / \alpha$ \;
				\If{$\rho/\alpha < \rho^*/\alpha^*$}{
					$(\alpha^*, \rho^*, c^*) \gets (\alpha, \rho, c)$ \;
				}
			}
		}
		{\color{winered}\tcc{Temporary increase of the budgets.}}
		$\Delta b = \frac{\cost(c^*) - \alpha^{*} \cost(c^*)}{|\{v_i \in V \colon \rho^* u_i(c^*) \geq b_i > 0\}|}$ \;
		$b_i^* = b_i + \max(0, \Delta b - \mathrm{over}_i) \text{~for each~} v_i \in V$ \;
		{\color{winered}\tcc{Finding the optimal candidate is the same as in the Method of Equal Shares.}}
		\For{$c \in C \setminus W$}{
			\eIf{$\sum_{v_i \in V : u_i(c) > 0} b_i^* < \cost(c)$}{
				$\rho_{\mathrm{MES}}(c) \gets \infty$ {\color{winered}\tcc*[l]{Project not affordable}}
			}{
				Let $v_{i_1}, \dots, v_{i_t}$ be a list of all voters $v_{i_j} \in V$ with $u_{i_j}(c) > 0$, ordered so that $b_{i_1}^*/u_{i_1}(c) \le \cdots \le b_{i_t}^*/u_{i_t}(c) $.\;
				\For{$s = 1, \dots, t$}{
	%				// Does voter $v_i_s$ need to spend all money with this $\rho$?\;
					$\rho_{\mathrm{MES}} \gets (\cost(c) - (b_{i_1}^* + \cdots + b_{i_{s-1}}^*))/(u_{i_{s}}(c) + \cdots + u_{i_t}(c))$\;
					\If{$\rho_{\mathrm{MES}}(c) \cdot u_{i_s} \le b_{i_s}^*$}{
						\textbf{break} {\color{winered}\tcc*[l]{we have found the optimal $\rho_{\mathrm{MES}}$-value}}
					}
				}
			}
		}
		$c^* \gets \argmin_{c \in C \setminus W} \rho_{\mathrm{MES}}(c)$ \;
		$W \gets W \cup \{c^*\}$ \;
		{\color{winered}\tcc{Updating budgets.}}
		\For{$v_i \in V$}{
			$\mathrm{pay}_i \gets \min(b_{i}^*, \rho_{\mathrm{MES}}(c^*) \cdot u_i(c^*))$ \;
			\eIf{$b_{i} \geq  \mathrm{pay}_i$}{
				$b_i \gets  b_i - \mathrm{pay}_i$ \;
			}{
				$\mathrm{over}_i \gets \mathrm{over}_i + (\mathrm{pay}_i - b_i)$ \;
				$b_i \gets 0$ \;
			}
		}
	}
	\Return{W}\;
\end{algorithm}

\clearpage
\section{Additional Empirical Results}
\label{app:experiments}

In \Cref{tab:pabulib:aggregated:app},
we report the aggregated statistics for each of our rules
based on their performance in the \pabulib{} real-world instances.
Additionally, we present the values of these statistics
when limited to the instances with a particular type of ballots used for voting.
\pabulib{} distinguishes four such ballot types:
\myemph{approval ballots}, where a voter indicates a subset of project he or she approves,
\myemph{choose-1 ballots}, in which a voter must select exactly one project,
\myemph{cumulative ballots}---a voter distributes a number of points between the projects, and
\myemph{ordinal ballots}, where a voter provides an ordering of (a subset of) projects---we use \myemph{Borda scores} to transform them to utilities.   

In \Cref{tab:experiments:pabulib:values},
we report the mean values of statistics presented in the first three plots of \Cref{fig:experiments:pabulib} in \Cref{sec:experiments}
together with standard deviations and five quantiles.
In the same table, we also present the $p$-values
for the significance of the differences between the average values of statistics for each pair of rules.
We note that they are very small, rarely exceeding $0.05$,
which means that almost all the observed differences
between the behavior of the rules are statistically significant.
In particular, we note that BOS has a lower average exclusion ratio than the original Method of Equal Shares
for every instance size
and the difference is significant at the 0.01 level.

Finally, in \Cref{fig:experiments:euclidean:app},
we present the full results of the experiment on the Euclidean elections.
Recall that we have defined the utility of voter $ v_i $ for candidate $ c $ 
as $ u_i(c) = (\dist(v_i, c) + \lambda)^{-1} $,
where $ \dist(v_i, c) $ is the Euclidean distance between them, and
the denominator is shifted by the constant of $\lambda$ to bound the maximal utility.
In the main body of the paper, we presented results for $\lambda = 1$.
Here, we also include results for $\lambda = 2$ and $\lambda = 0.5$.
Note that taking $\lambda=2$ instead of $\lambda=1$
leads to exactly the same outcomes
as multiplying all distances by $0.5$
since all of the considered rules are invariant towards
scaling all the utilities by a constant.
Similarly, taking $\lambda=0.5$ instead of $\lambda=1$
is equivalent to doubling the distances.

\begin{table}
	\centering
	\begin{tabular}{l ccccc}
		\toprule
		Metric & Util. & Eq. Shares & BOS & BOS+ & FrES \\
		\midrule
		Avg. rel. cost satisfaction & 1.000 & 0.836 & 0.903 & 0.909 & 0.921 \\ 
		\midrule
		\rightcell{\quad instances with approval ballots} & 1.000 & 0.846 & 0.890 & 0.900 & 0.960 \\ 
		\rightcell{\quad instances with choose-1 ballots} & 1.000 & 0.908 & 0.984 & 0.984 & 0.825 \\ 
		\rightcell{\quad instances with cumulative ballots} & 1.000 & 0.751 & 0.890 & 0.887 & 0.885 \\ 
		\rightcell{\quad instances with ordinal ballots} & 1.000 & 0.821 & 0.896 & 0.902 & 0.885 \\ 
		\midrule
		Avg. rel. score satisfaction & 1.000 & 1.201 & 1.160 & 1.169 & 1.029 \\ 
		\midrule
		\rightcell{\quad instances with approval ballots} & 1.000 & 1.286 & 1.201 & 1.215 & 1.052 \\ 
		\rightcell{\quad instances with choose-1 ballots} & 1.000 & 0.964 & 1.004 & 1.004 & 0.855 \\ 
		\rightcell{\quad instances with cumulative ballots} & 1.000 & 1.099 & 1.138 & 1.140 & 1.070 \\ 
		\rightcell{\quad instances with ordinal ballots} & 1.000 & 1.187 & 1.164 & 1.171 & 1.060 \\ 
		\midrule
		Avg. exclusion ratio & 19.85\% & 17.62\% & 16.16\% & 16.50\% & 0.00\% \\ 
		\midrule
		\rightcell{\quad instances with approval ballots} & 13.92\% & 11.86\% & 10.36\% & 10.88\% & 0.00\% \\ 
		\rightcell{\quad instances with choose-1 ballots} & 41.36\% & 43.37\% & 41.19\% & 41.19\% & 0.00\% \\ 
		\rightcell{\quad instances with cumulative ballots} & 29.57\% & 27.12\% & 24.90\% & 24.89\% & 0.00\% \\ 
		\rightcell{\quad instances with ordinal ballots} & 12.38\% & 5.10\% & 5.54\% & 5.85\% & 0.00\% \\ 
		\midrule
		Avg. running time in sec. & 0.001 & 6.822 & 0.086 & 0.263 & 2.151 \\ 
		\midrule
		\rightcell{\quad instances with approval ballots} & 0.001 & 10.585 & 0.120 & 0.382 & 3.296 \\ 
		\rightcell{\quad instances with choose-1 ballots} & 0.000 & 1.087 & 0.044 & 0.116 & 0.028 \\ 
		\rightcell{\quad instances with cumulative ballots} & 0.000 & 0.751 & 0.014 & 0.031 & 0.060 \\ 
		\rightcell{\quad instances with ordinal ballots} & 0.000 & 2.914 & 0.060 & 0.151 & 1.676 \\ 
		\midrule
		Avg. EJR+ violations & 0.953 & 0.000 & 0.061 & 0.060 & 0.000 \\ 
		\midrule
		\rightcell{\quad instances with approval ballots} & 1.138 & 0.000 & 0.059 & 0.058 & 0.000 \\ 
		\rightcell{\quad instances with choose-1 ballots} & 0.141 & 0.000 & 0.071 & 0.071 & 0.000 \\ 
		\midrule
		EJR+ violation instances & 26.20\% & 0.00\% & 4.59\% & 4.48\% & 0.00\% \\ 
		\midrule
		\rightcell{\quad instances with approval ballots} & 29.62\% & 0.00\% & 4.29\% & 4.16\% & 0.00\% \\ 
		\rightcell{\quad instances with choose-1 ballots} & 11.18\% & 0.00\% & 5.88\% & 5.88\% & 0.00\% \\ 
		\midrule
		Avg. budget spending & 96.47\% & $44.84\%^{\dagger}$ & 93.98\% & 93.95\% & $93.40\%^{\dagger}$ \\ 
		\midrule
		\rightcell{\quad instances with approval ballots} & 96.09\% & $53.86\%^{\dagger}$ & 93.20\% & 93.16\% & $92.51\%^{\dagger}$ \\ 
		\rightcell{\quad instances with choose-1 ballots} & 97.40\% & $14.98\%^{\dagger}$ & 96.35\% & 96.35\% & $95.94\%^{\dagger}$ \\ 
		\rightcell{\quad instances with cumulative ballots} & 95.79\% & $27.08\%^{\dagger}$ & 92.77\% & 92.77\% & $90.90\%^{\dagger}$ \\ 
		\rightcell{\quad instances with ordinal ballots} & 98.16\% & $56.85\%^{\dagger}$ & 96.64\% & 96.59\% & $98.02\%^{\dagger}$ \\ 
		\midrule
		Exhausted budgets & 100.00\% & $6.99\%^{\dagger}$ & 99.84\% & 99.69\% & $24.65\%^{\dagger}$ \\ 
		\midrule
		\rightcell{\quad instances with approval ballots} & 100.00\% & $8.31\%^{\dagger}$ & 99.73\% & 99.60\% & $17.29\%^{\dagger}$ \\ 
		\rightcell{\quad instances with choose-1 ballots} & 100.00\% & $4.12\%^{\dagger}$ & 100.00\% & 100.00\% & $65.88\%^{\dagger}$ \\ 
		\rightcell{\quad instances with cumulative ballots} & 100.00\% & $4.00\%^{\dagger}$ & 100.00\% & 99.50\% & $16.00\%^{\dagger}$ \\ 
		\rightcell{\quad instances with ordinal ballots} & 100.00\% & $7.59\%^{\dagger}$ & 100.00\% & 100.00\% & $25.95\%^{\dagger}$ \\ 
		\bottomrule
	\end{tabular}
	\caption{Aggregated statistics from running our rules on instances from \pabulib.
The values for Equal Shares assume Add1U completion, except for the average budget spending and exhausted budgets.
Similarly, FrES is completed in a utilitarian fashion except for these two cases. The usage of a completion method is denoted by $\dagger$.}
	\label{tab:pabulib:aggregated:app}
\end{table}

\begin{table*}[p!]
	\small
	\renewcommand{\arraystretch}{0.97}
	\setlength{\tabcolsep}{4pt}
	\centering
	\begin{tabular}{cc cccc cccc cccc}
		\toprule
		\multicolumn{2}{c}{Statistics} & \multicolumn{4}{c}{Score satisfaction} & \multicolumn{4}{c}{Cost satisfaction} & \multicolumn{4}{c}{Exclusion ratio} \\
		\multicolumn{2}{c}{No. projects} & $1$--$8$ & $9$--$16$ & $17$--$28$ & $29$+ & $1$--$8$ & $9$--$16$ & $17$--$28$ & $29$+ & $1$--$8$ & $9$--$16$ & $17$--$28$ & $29$+ \\
		\midrule
			& mean & $1.000$ & $1.000$ & $1.000$ & $1.000$ & $1.000$ & $1.000$ & $1.000$ & $1.000$ & $0.206$ & $0.209$ & $0.178$ & $0.199$ \\ 
			& std & $0.000$ & $0.000$ & $0.000$ & $0.000$ & $0.000$ & $0.000$ & $0.000$ & $0.000$ & $0.201$ & $0.201$ & $0.185$ & $0.166$ \\ 
			& q10 & $1.000$ & $1.000$ & $1.000$ & $1.000$ & $1.000$ & $1.000$ & $1.000$ & $1.000$ & $0.000$ & $0.013$ & $0.006$ & $0.050$ \\ 
		Util. & q25 & $1.000$ & $1.000$ & $1.000$ & $1.000$ & $1.000$ & $1.000$ & $1.000$ & $1.000$ & $0.025$ & $0.051$ & $0.040$ & $0.086$ \\ 
			& q50 & $1.000$ & $1.000$ & $1.000$ & $1.000$ & $1.000$ & $1.000$ & $1.000$ & $1.000$ & $0.148$ & $0.145$ & $0.121$ & $0.147$ \\ 
			& q75 & $1.000$ & $1.000$ & $1.000$ & $1.000$ & $1.000$ & $1.000$ & $1.000$ & $1.000$ & $0.324$ & $0.304$ & $0.239$ & $0.252$ \\ 
			& q90 & $1.000$ & $1.000$ & $1.000$ & $1.000$ & $1.000$ & $1.000$ & $1.000$ & $1.000$ & $0.526$ & $0.523$ & $0.457$ & $0.414$ \\ 
		\midrule
			& mean & $1.061$ & $1.148$ & $1.265$ & $1.350$ & $0.820$ & $0.845$ & $0.834$ & $0.848$ & $0.229$ & $0.200$ & $0.134$ & $0.135$ \\ 
			& std & $0.463$ & $0.620$ & $0.625$ & $0.635$ & $0.302$ & $0.230$ & $0.177$ & $0.123$ & $0.227$ & $0.204$ & $0.155$ & $0.155$ \\ 
			& q10 & $0.860$ & $0.955$ & $0.996$ & $1.038$ & $0.264$ & $0.429$ & $0.575$ & $0.671$ & $0.000$ & $0.005$ & $0.002$ & $0.019$ \\ 
		MES & q25 & $1.000$ & $1.000$ & $1.000$ & $1.114$ & $0.770$ & $0.777$ & $0.735$ & $0.816$ & $0.022$ & $0.047$ & $0.019$ & $0.041$ \\ 
			& q50 & $1.000$ & $1.000$ & $1.115$ & $1.257$ & $1.000$ & $0.970$ & $0.890$ & $0.886$ & $0.164$ & $0.136$ & $0.072$ & $0.081$ \\ 
			& q75 & $1.000$ & $1.075$ & $1.281$ & $1.475$ & $1.000$ & $1.000$ & $0.970$ & $0.927$ & $0.395$ & $0.291$ & $0.190$ & $0.159$ \\ 
			& q90 & $1.114$ & $1.324$ & $1.573$ & $1.682$ & $1.000$ & $1.000$ & $1.000$ & $0.960$ & $0.573$ & $0.500$ & $0.398$ & $0.302$ \\ 
		\midrule
			& mean & $1.024$ & $1.085$ & $1.240$ & $1.314$ & $0.941$ & $0.924$ & $0.878$ & $0.866$ & $0.201$ & $0.186$ & $0.126$ & $0.127$ \\ 
			& std & $0.148$ & $0.300$ & $0.579$ & $0.632$ & $0.162$ & $0.146$ & $0.145$ & $0.103$ & $0.200$ & $0.191$ & $0.149$ & $0.155$ \\ 
			& q10 & $0.996$ & $0.987$ & $0.999$ & $1.005$ & $0.793$ & $0.739$ & $0.678$ & $0.755$ & $0.000$ & $0.005$ & $0.002$ & $0.018$ \\ 
		BOS & q25 & $1.000$ & $1.000$ & $1.000$ & $1.086$ & $1.000$ & $0.908$ & $0.830$ & $0.813$ & $0.018$ & $0.046$ & $0.020$ & $0.035$ \\ 
			& q50 & $1.000$ & $1.000$ & $1.095$ & $1.224$ & $1.000$ & $1.000$ & $0.920$ & $0.883$ & $0.147$ & $0.122$ & $0.072$ & $0.069$ \\ 
			& q75 & $1.000$ & $1.033$ & $1.261$ & $1.411$ & $1.000$ & $1.000$ & $0.992$ & $0.941$ & $0.319$ & $0.267$ & $0.171$ & $0.150$ \\ 
			& q90 & $1.007$ & $1.265$ & $1.538$ & $1.612$ & $1.000$ & $1.000$ & $1.000$ & $0.978$ & $0.525$ & $0.476$ & $0.366$ & $0.283$ \\ 
		\midrule
			& mean & $1.025$ & $1.088$ & $1.249$ & $1.337$ & $0.941$ & $0.926$ & $0.881$ & $0.885$ & $0.201$ & $0.187$ & $0.129$ & $0.137$ \\ 
			& std & $0.148$ & $0.300$ & $0.590$ & $0.634$ & $0.163$ & $0.146$ & $0.147$ & $0.103$ & $0.199$ & $0.191$ & $0.149$ & $0.153$ \\ 
			& q10 & $1.000$ & $0.999$ & $1.000$ & $1.023$ & $0.793$ & $0.731$ & $0.678$ & $0.781$ & $0.000$ & $0.007$ & $0.002$ & $0.022$ \\ 
		BOS+ & q25 & $1.000$ & $1.000$ & $1.000$ & $1.101$ & $1.000$ & $0.912$ & $0.831$ & $0.849$ & $0.021$ & $0.047$ & $0.022$ & $0.042$ \\ 
			& q50 & $1.000$ & $1.000$ & $1.104$ & $1.239$ & $1.000$ & $1.000$ & $0.923$ & $0.905$ & $0.147$ & $0.124$ & $0.073$ & $0.081$ \\ 
			& q75 & $1.000$ & $1.038$ & $1.268$ & $1.433$ & $1.000$ & $1.000$ & $1.000$ & $0.956$ & $0.319$ & $0.267$ & $0.174$ & $0.164$ \\ 
			& q90 & $1.007$ & $1.265$ & $1.549$ & $1.667$ & $1.000$ & $1.000$ & $1.000$ & $0.985$ & $0.525$ & $0.476$ & $0.372$ & $0.304$ \\ 
		\midrule
			& mean & $1.037$ & $0.996$ & $1.048$ & $1.038$ & $0.973$ & $0.943$ & $0.904$ & $0.857$ & $0.000$ & $0.000$ & $0.000$ & $0.000$ \\ 
			& std & $0.255$ & $0.272$ & $0.442$ & $0.378$ & $0.150$ & $0.162$ & $0.144$ & $0.104$ & $0.000$ & $0.000$ & $0.000$ & $0.000$ \\ 
			& q10 & $0.812$ & $0.790$ & $0.802$ & $0.772$ & $0.781$ & $0.731$ & $0.749$ & $0.713$ & $0.000$ & $0.000$ & $0.000$ & $0.000$ \\ 
		FrES & q25 & $0.943$ & $0.886$ & $0.894$ & $0.902$ & $0.910$ & $0.875$ & $0.837$ & $0.814$ & $0.000$ & $0.000$ & $0.000$ & $0.000$ \\ 
			& q50 & $1.000$ & $0.970$ & $0.977$ & $1.014$ & $0.990$ & $0.966$ & $0.920$ & $0.881$ & $0.000$ & $0.000$ & $0.000$ & $0.000$ \\ 
			& q75 & $1.074$ & $1.027$ & $1.084$ & $1.126$ & $1.006$ & $1.011$ & $0.982$ & $0.933$ & $0.000$ & $0.000$ & $0.000$ & $0.000$ \\ 
			& q90 & $1.236$ & $1.222$ & $1.282$ & $1.296$ & $1.127$ & $1.071$ & $1.016$ & $0.960$ & $0.000$ & $0.000$ & $0.000$ & $0.000$ \\ 
		\midrule
		\multicolumn{2}{c}{Util. vs. MES} & $0.008$ & $0.000$ & $0.000$ & $0.000$ & $0.000$ & $0.000$ & $0.000$ & $0.000$ & $0.000$ & $0.017$ & $0.000$ & $0.000$ \\ 
		\multicolumn{2}{c}{Util. vs. BOS} & $0.001$ & $0.000$ & $0.000$ & $0.000$ & $0.000$ & $0.000$ & $0.000$ & $0.000$ & $0.015$ & $0.000$ & $0.000$ & $0.000$ \\ 
		\multicolumn{2}{c}{Util. vs. BOS+} & $0.001$ & $0.000$ & $0.000$ & $0.000$ & $0.000$ & $0.000$ & $0.000$ & $0.000$ & $0.028$ & $0.000$ & $0.000$ & $0.000$ \\ 
		\multicolumn{2}{c}{Util. vs. FrES} & $0.004$ & $0.385$ & $0.029$ & $0.041$ & $0.000$ & $0.000$ & $0.000$ & $0.000$ & $0.000$ & $0.000$ & $0.000$ & $0.000$ \\ 
		\multicolumn{2}{c}{MES vs. BOS} & $0.064$ & $0.019$ & $0.024$ & $0.000$ & $0.000$ & $0.000$ & $0.000$ & $0.001$ & $0.000$ & $0.000$ & $0.000$ & $0.000$ \\ 
		\multicolumn{2}{c}{MES vs. BOS+} & $0.066$ & $0.024$ & $0.107$ & $0.000$ & $0.000$ & $0.000$ & $0.000$ & $0.000$ & $0.000$ & $0.000$ & $0.001$ & $0.072$ \\ 
		\multicolumn{2}{c}{MES vs. FrES} & $0.133$ & $0.000$ & $0.000$ & $0.000$ & $0.000$ & $0.000$ & $0.000$ & $0.055$ & $0.000$ & $0.000$ & $0.000$ & $0.000$ \\ 
		\multicolumn{2}{c}{BOS vs. BOS+} & $0.224$ & $0.000$ & $0.000$ & $0.000$ & $0.441$ & $0.076$ & $0.010$ & $0.000$ & $0.032$ & $0.000$ & $0.000$ & $0.000$ \\ 
		\multicolumn{2}{c}{BOS vs. FrES} & $0.139$ & $0.000$ & $0.000$ & $0.000$ & $0.004$ & $0.036$ & $0.003$ & $0.115$ & $0.000$ & $0.000$ & $0.000$ & $0.000$ \\ 
		\multicolumn{2}{c}{BOS+ vs. FrES} & $0.148$ & $0.000$ & $0.000$ & $0.000$ & $0.004$ & $0.053$ & $0.008$ & $0.000$ & $0.000$ & $0.000$ & $0.000$ & $0.000$ \\ 
		\bottomrule
	\end{tabular}
	\caption{Detailed values of statistics from running our rules on instances from Pabulib.
		In the first 35 rows, for each rule, statistic, and size range,
		we provide the average, standard deviation, median, 1st and 3rd quartile, as well as 10th and 90th centile
		of the observed values of the statistic.
		In the last 10 rows, we present $p$-values for
		the significance of the difference between the average values for a given pair of rules.}
	\label{tab:experiments:pabulib:values}
\end{table*}

\begin{figure*}[t]
	\centering
	\setlength{\tabcolsep}{-1pt}
	\renewcommand{\arraystretch}{0}
	\begin{tabular}{c cccccc}
		% & \multicolumn{3}{c}{Utility model with $\lambda=1$} 
		% & \multicolumn{3}{c}{Utility model with $\lambda=2$} \\
		& {\small Voters} & {\small Utilitarian} & {\small Equal Shares} & {\small BOS} & {\small BOS+} & {\small FrES} \\
		\multirow{4}{*}{\rotatebox{90}{Utility model with $\lambda=0.5$\quad\quad\quad\ }\ } & & & & & & \\ &
		\includegraphics[width=2.1cm]{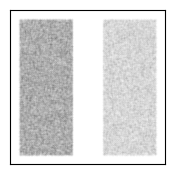} &
		\includegraphics[width=2.1cm]{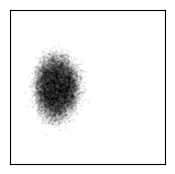} &
		\includegraphics[width=2.1cm]{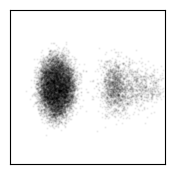} &
		\includegraphics[width=2.1cm]{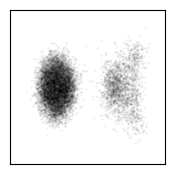} &
		\includegraphics[width=2.1cm]{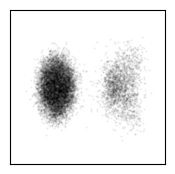} &
		\includegraphics[width=2.1cm]{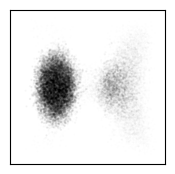} \\ &
		\includegraphics[width=2.1cm]{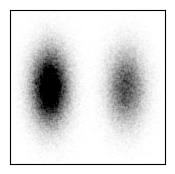} &
		\includegraphics[width=2.1cm]{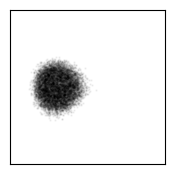} &
		\includegraphics[width=2.1cm]{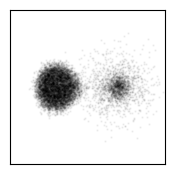} &
		\includegraphics[width=2.1cm]{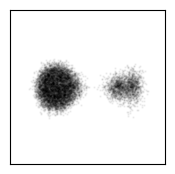} &
		\includegraphics[width=2.1cm]{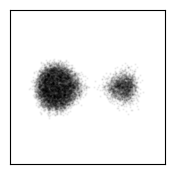} &
		\includegraphics[width=2.1cm]{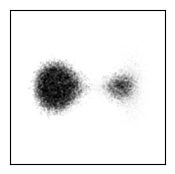} \\ &
		\includegraphics[width=2.1cm]{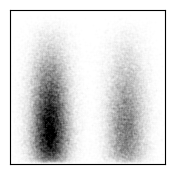} &
		\includegraphics[width=2.1cm]{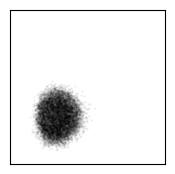} &
		\includegraphics[width=2.1cm]{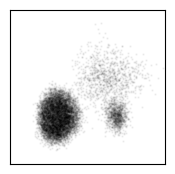} &
		\includegraphics[width=2.1cm]{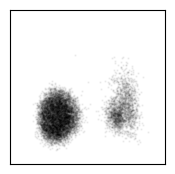} &
		\includegraphics[width=2.1cm]{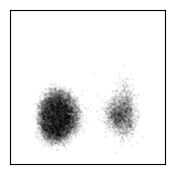} &
		\includegraphics[width=2.1cm]{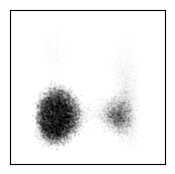} \\
		\multirow{4}{*}{\rotatebox{90}{Utility model with $\lambda=1$\quad\quad\quad\quad }\ } & \vspace{0.1cm} & & & & & \\ &
		\includegraphics[width=2.1cm]{img/euclidean/1_0_voters.png} &
		\includegraphics[width=2.1cm]{img/euclidean/1_0_util.png} &
		\includegraphics[width=2.1cm]{img/euclidean/1_0_mes.png} &
		\includegraphics[width=2.1cm]{img/euclidean/1_0_bos.png} &
		\includegraphics[width=2.1cm]{img/euclidean/1_0_bos_plus.png} &
		\includegraphics[width=2.1cm]{img/euclidean/1_0_fres.png} \\ &
		\includegraphics[width=2.1cm]{img/euclidean/1_1_voters.png} &
		\includegraphics[width=2.1cm]{img/euclidean/1_1_util.png} &
		\includegraphics[width=2.1cm]{img/euclidean/1_1_mes.png} &
		\includegraphics[width=2.1cm]{img/euclidean/1_1_bos.png} &
		\includegraphics[width=2.1cm]{img/euclidean/1_1_bos_plus.png} &
		\includegraphics[width=2.1cm]{img/euclidean/1_1_fres.png} \\ &
		\includegraphics[width=2.1cm]{img/euclidean/1_2_voters.png} &
		\includegraphics[width=2.1cm]{img/euclidean/1_2_util.png} &
		\includegraphics[width=2.1cm]{img/euclidean/1_2_mes.png} &
		\includegraphics[width=2.1cm]{img/euclidean/1_2_bos.png} &
		\includegraphics[width=2.1cm]{img/euclidean/1_2_bos_plus.png} &
		\includegraphics[width=2.1cm]{img/euclidean/1_2_fres.png} \\
		\multirow{4}{*}{\rotatebox{90}{Utility model with $\lambda=2$\quad\quad\quad\quad }\ } & \vspace{0.1cm} & & & & & \\ &
		\includegraphics[width=2.1cm]{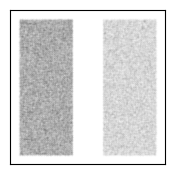} &
		\includegraphics[width=2.1cm]{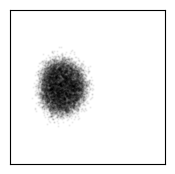} &
		\includegraphics[width=2.1cm]{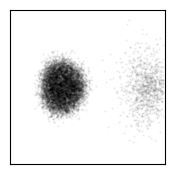} &
		\includegraphics[width=2.1cm]{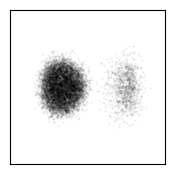} &
		\includegraphics[width=2.1cm]{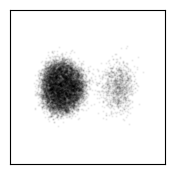} &
		\includegraphics[width=2.1cm]{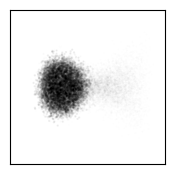} \\ &
		\includegraphics[width=2.1cm]{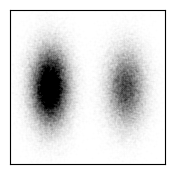} &
		\includegraphics[width=2.1cm]{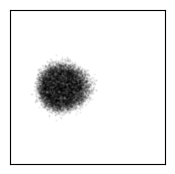} &
		\includegraphics[width=2.1cm]{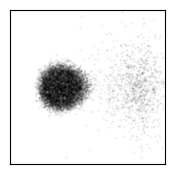} &
		\includegraphics[width=2.1cm]{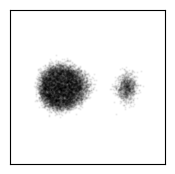} &
		\includegraphics[width=2.1cm]{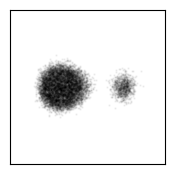} &
		\includegraphics[width=2.1cm]{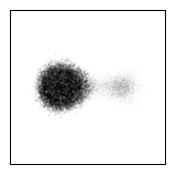} \\ &
		\includegraphics[width=2.1cm]{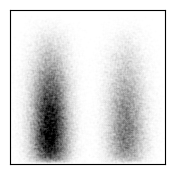} &
		\includegraphics[width=2.1cm]{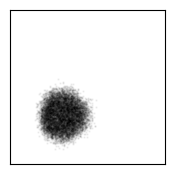} &
		\includegraphics[width=2.1cm]{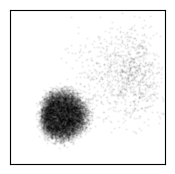} &
		\includegraphics[width=2.1cm]{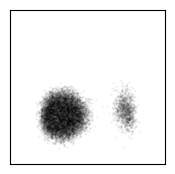} &
		\includegraphics[width=2.1cm]{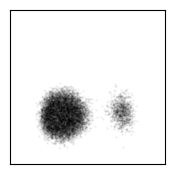} &
		\includegraphics[width=2.1cm]{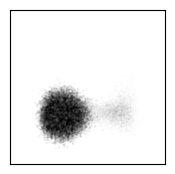} \\
	\end{tabular}
	\caption{Results of the experiment on Euclidean elections.
		The first column presents the superimposed positions of voters in all 1000 generated elections.
		The following columns show positions of candidates selected by respective rules
		(for FrES, the opacity of each point is proportional to the selected fraction of the respective candidate).}
	\label{fig:experiments:euclidean:app}
\end{figure*}

\end{document}